\documentclass[10pt,journal,compsoc]{IEEEtran}

\ifCLASSOPTIONcompsoc
  \usepackage[nocompress]{cite}
\else
  \usepackage{cite}
\fi

\ifCLASSINFOpdf
   \usepackage[pdftex]{graphicx}
 
\else
 
\fi

\usepackage{fixltx2e}
\usepackage{amsmath,amsfonts,amssymb,amsthm}
\newtheorem{theorem}{Theorem}

\newtheorem{definition}{Definition}

\newtheorem{proposition}{Proposition}

\usepackage[noend]{algpseudocode}
\usepackage{algorithmicx,algorithm}

\usepackage{array}

\usepackage{subfigure} 

\usepackage{booktabs}
\usepackage{multirow}
\usepackage{makecell}
\usepackage{diagbox}

\usepackage{url}

\begin{document}

\title{A Lattice-Based Embedding Method for Reversible Audio Watermarking}

\author{Junren Qin, Shanxiang Lyu, Jiarui Deng,  Xingyuan Liang, Shijun Xiang, Hao Chen
\IEEEcompsocitemizethanks{\IEEEcompsocthanksitem The authors are with the College of Information Science and Technology,  and the College of Cyber Security, Jinan University, Guangzhou 510632, China.
Corresponding author: Shanxiang Lyu (E-mail: lsx07@jnu.edu.cn).
}
}


\IEEEtitleabstractindextext{%
\begin{abstract}
Reversible audio watermarking (RAW) is a promising technique in various applications. To simultaneously meet the demand of achieving high imperceptibility and robustness,  this paper proposes a novel RAW scheme based on lattices. The scheme is referred to as  Meet-in-the-Middle Embedding (MME), in which the lattice quantization errors are properly scaled and added back to the quantized host signals. Simulations show that MME excels in a wide range of metrics including signal-to-watermark ratio (SWR), objective difference grade (ODG), and bit error rate (BER).
\end{abstract}

}

\maketitle

\IEEEdisplaynontitleabstractindextext
\IEEEpeerreviewmaketitle

\IEEEraisesectionheading{\section{Introduction}\label{sec:introduction}}

\IEEEPARstart{R}{eversible} 
data  hiding (RDH) is a technique where a payload is embedded into host data such that the consistency of the host is perfectly preserved and the host data are restored during the extraction of the payload \cite{shi2016reversible}. Nowadays, audio is widely spread on the Internet and has been one of the main digital multimedia carriers. Its associated Reversible Audio Watermarking (RAW) techniques are beneficial for archiving, transmitting, and authenticating high-quality audio data that contains metadata, secret data, and so forth \cite{DBLP:journals/sigpro/HuaHSGT16}.

Similar to other types of multimedia, a critical step of RAW is to design a proper reversible function. 
Varying on how the reversible function is designed,
conventional RDH methods include difference expansion (DE) \cite{1227616}, prediction-error expansion (PEE) \cite{1421361,WU2014387}, histogram shifting (HS) \cite{ni2006reversible}. A common feature of these methods is to embed some clues of the cover object into the watermarked object, such that the cover object can be restored. Unfortunately, 
these methods are mostly only suitable for a finite-alphabet cover domain (e.g., the domain of $\lbrace0,...,255\rbrace$ in a gray-scale pixel), which implies that they are more suitable for images rather than audio.

In recent years, the celebrated quantization index modulation (QIM) \cite{chen2001quantization} has also been adapted for RDH \cite{ko2012nested,peng2011effective,9328388}, in which the cover domain more generally includes the real-valued domain where the digital audio resides.
Ko \textit{et al.} \cite{ko2012nested} developed a reversible watermarking algorithm based on nested QIM via the iterative process that reduces the normalized error. Peng \emph{et al.} \cite{peng2011effective} made an improvement in RDH for one dimensional scalar-QIM by maintaining the relative distance, and extended it into $N$-dimensional region nesting model \cite{9328388}. These methods are collectively referred to as improved-QIM (IQIM), as is named in \cite{peng2011effective}. It is noteworthy that applying IQIM to digital audio is straightforward, although they generally fail to meet the demand of high robustness and imperceptibility. The reason of no robustness is that the scaling factor of IQIM is fixed as $1-\frac{1}{2^b}$. The reason of high distortion is more involved: i) The constructed difference vector is rather sub-optimal. ii) The embedding method is based on the integer lattice $\mathbb{Z}^n$, whose distortion performance is worse than other optimal lattices.

The robustness of the embedded message is important in the sense that
the cover object  may go through some attacks such as additive noises. In this regard, robust RAW has been investigated in \cite{nishimura2013reversible,LIANG2020107584,nishimura2016reversible}. To be concise,
Nishimura \textit{et al.} \cite{nishimura2013reversible} proposed a robust RAW scheme by combining quantization index modulation (QIM) and amplitude expansion, and also a scheme based on error-expansion of linear prediction \cite{nishimura2016reversible}  for segmental audio and G.711 speech. Liang \textit{et al.} \cite{LIANG2020107584} proposed a robust RAW scheme based on high-order difference statistics.  Nevertheless, these works overemphasize the robustness of embedded messages, but ignore or even sacrifice its imperceptibility. 
Since the sensibility of human auditory system (HAS),
the watermarked cover object should not cause 
perceivable changes.

To address the above issues, this paper makes a retrospective and prospective revisit to RDH from the perspective of lattices \cite{zamir2014lattice,DBLP:journals/tit/LingB14,DBLP:journals/tit/LyuC019}, with the hope of designing an efficient RAW scheme that simultaneously achieves good robustness and imperceptibility. Lattices are discrete additive subgroups of the Euclidean vector space, and QIM is in essence quantization based on nested lattice codes \cite{moulin2005data,zamir2014lattice}. The reversibility, robustness and imperceptibility features of QIM can then be better accommodated by using more advanced lattice coding skills. Specifically, the contributions of this paper, along with the highlights, are summarized as follows.

\begin{itemize}
	\item First, we propose a Meet-in-the-Middle Embedding (MME) RAW scheme. Within the QIM framework, there exists a difference vector between the 
	cover vector and the quantized cover vector. By properly scaling the difference vector and adding it back to the quantized cover vector, it becomes possible for the receiver to estimate the cover vector with the aid of the scaled difference vector. As the reversibility of MME hinges on the properly choosing the scaling factor, we derive the feasible range of the scaling factor.
	\item Second, we analyze the performance of MME in terms of distortion and robustness.
	More specifically, we compare the signal-to-watermark ratio (SWR) and generalized signal-to-noise ratio (GSNR) between MME and IQIM \cite{peng2011effective}. 
	\item Finally, we make a comprehensive comparison between MME and other state-of-the-art RAW methods, and the   
	numerical simulations show that MME excels in the robustness and imperceptibility metrics such as 
	SWR, GSNR, objective difference grade (ODG), and bit error rate (BER). The simulated results ideally align with our theoretical analyses.
\end{itemize}

Notation: Matrices and column vectors are denoted by
uppercase and lowercase boldface letters. We write $\lfloor \cdot \rfloor$ for the floor function, and $|\cdot|$ for the  absolute value.

\section{Preliminaries}
\subsection{QIM}\label{Sec.QIM}
The one-bit scalar-QIM is perhaps the most popular version of QIM in the data hiding community \cite{moulin2005data}. Its rationale can be explained by a one-bit embedding example in Fig. \ref{Fig. QIM embed model}. Denote circle and cross positions in Fig. \ref{Fig. QIM embed model} as two sets $\Lambda_0$ and $\Lambda_1$, respectively.
Given a host/cover sample $s\in\mathbb{R}$ and a one-bit message $m \in \{0,1\}$, the watermarked value is simply moving $s$ to the nearest point in  $\Lambda_0$ when $m=0$, and to the nearest point in  $\Lambda_1$ when $m=1$.


\begin{figure}[t!]
	\centering
	\includegraphics[width=0.4\textwidth]{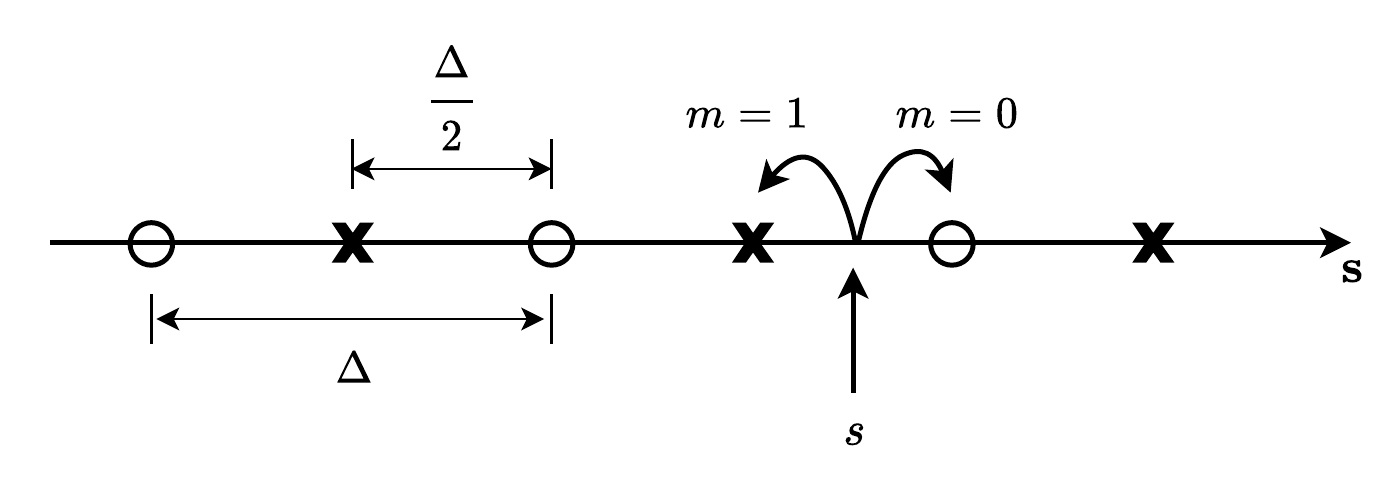} 
	\caption{Embedding one bit into a sample with original QIM.} 
	\label{Fig. QIM embed model}
\end{figure}

Define $Q(s)=\Delta\lfloor s/\Delta\rfloor$ with $\Delta$ being a step-size parameter.
Then the embedding process can be described as 
\begin{equation}\label{Eq. conventional QIM embedding}
	s_{\mathrm{QIM}}\triangleq Q_m(s)=Q(s-d_m)+d_m,~m\in \{0,1\},
\end{equation} 
where $d_0=-(\Delta/4)$, $d_1=\Delta/4$,  $\Lambda_0=d_0+\Delta\mathbb{Z}$ and $\Lambda_1=d_1+\Delta\mathbb{Z}$. 

Assume that the transmitted $s_{\mathrm{QIM}}$ has undergone the contamination of an additive noise term $n$, then at the receiver's side we have: $y= s_{\mathrm{QIM}}+n$. A minimum distance decoder is therefore given by
\begin{equation}\label{Eq. QIM extract}
	\hat{m}=\mathop{arg\min}_{m\in\{0,1\}}\left[\mathop{\min}_{s\in\Lambda_m}|y-s|\right].
\end{equation}
If $|n|<\Delta/4$, the $\hat{m}$ is correct.

Regarding the embedding distortion, as shown in Fig. \ref{Fig. QIM distortion}, the maximum error caused by embedding is $\Delta/2$. If the quantization errors are distributed uniformly over $[-(\Delta/2),(\Delta/2)]$, then the mean-squared embedding distortion is $D=\Delta^2/12$.
\begin{figure}[t!]
	\centering
	\includegraphics[width=0.5\textwidth]{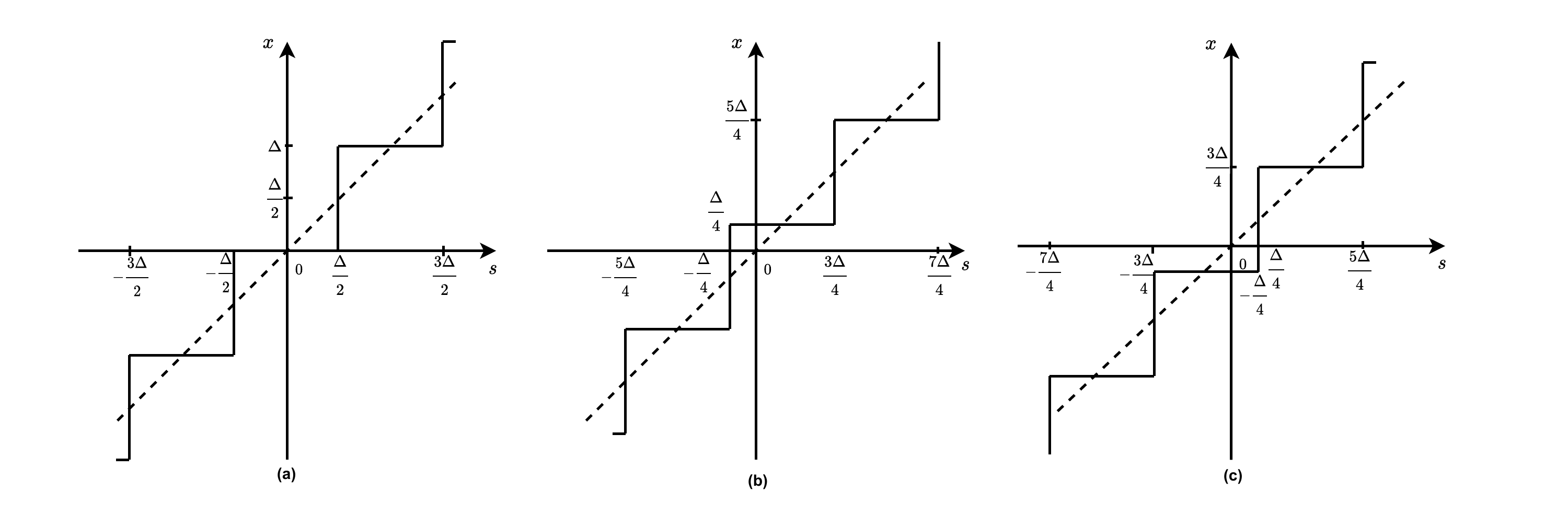} 
	\caption{Selection of watermarked signal $x$ with given $s$ and $m \in \{0,1\}$ in one-bit scalar-QIM. (a) Prototype symmetric function. (b) $x=Q_m(s)$ with $m = 0$. (c) $x=Q_m(s)$ with $m = 1$.} 
	\label{Fig. QIM distortion}
\end{figure}


\subsection{IQIM}\label{IQIM}

Inspired by conventional QIM, a reversible version of QIM called improved-QIM (IQIM) has been proposed in \cite{peng2011effective}. Its embedding and extraction processes can be described as follows.

\subsubsection{Embedding}
\noindent \textbf{Step i)} Calculate the number of quantization interval $\gamma$ and relative distance of $s$ in the $\gamma^{th}$ quantization interval $r$ by:
\begin{equation} \label{prepare step}
	\left\{
	\begin{aligned}
		\gamma&=\left\lfloor\frac{s}{2^b\Delta}\right\rfloor\\
		r&=s-2^b\gamma\Delta
	\end{aligned},\right.
\end{equation}
where $b$ is the bit length of a watermark $m$, and $\Delta$ is again the step size of quantization.

\noindent \textbf{Step ii)} To embed the watermark $m$, the final embedded point $s_{\mathrm{IQIM}}$ is set as
\begin{equation}\label{Eq. REF of IQIM}
	s_{\mathrm{IQIM}}=2^b\gamma\Delta+m \Delta+\frac{r}{2^b}.
\end{equation} 

\subsubsection{Extraction}
\noindent \textbf{Step i)} With either a clean or a noisy version of $s_{\mathrm{IQIM}}$ denoted as $y$, calculate the number of quantization interval $\gamma_w$ and the relative distance $r_w$ by
\begin{equation}
	\left\{
	\begin{aligned}
		\gamma_w&=\left\lfloor\frac{y}{\Delta}\right\rfloor\\
		r_w&=y-\gamma_w\Delta
	\end{aligned}
	.\right.
\end{equation}

\noindent \textbf{Step ii)} Estimate the watermark $m$ by
\begin{equation}
	\hat{m}=\gamma_w-2^b\lfloor\frac{\gamma_w}{2^b}\rfloor.
\end{equation}

\noindent \textbf{Step iii)} Estimate the   cover $s$ by
\begin{equation}
	\hat{s}=2^b\lfloor\frac{\gamma_w}{2^b}\rfloor\Delta+2^br_w.
\end{equation}

If $y$ is noiseless, one may verify that $\hat{m}=m$ and $\hat{s}=s$. For noisy $y$, although $\hat{s}$ becomes inaccurate, $\hat{m}$ may still be reliable if the noise pollution is small enough. In addition,
IQIM can also be described by an $N$-dimension integer set $\mathbb{Z}^N$, which is no more than independently performing embedding and extraction in each dimension. For instance, the embedding function w.r.t. the watermark $\mathbf{m}$ is given by
\begin{align}
	\mathbf{s}_{\mathrm{IQIM}}=\beta\left[2^b\Delta \left[\gamma_1~\cdots~\gamma_N\right]^\top+\frac{\Delta\mathbf{m}}{\beta}\right]+ (1-\beta)\mathbf{s},
\end{align}
where $\beta=1-\frac{1}{2^b}$ denotes the scaling factor in IQIM.

\subsection{Lattices}
This paper will use some concepts from lattices to arrive at simpler description of the algorithm and more elegant analysis. Hereby we review some basic concepts of lattices and nested lattice codes. 


An $N$-dimensional lattice in $\mathbb{R}^N$ is defined as  $\Lambda=\left\lbrace \mathbf{Gz}|\mathbf{z}\in \mathbb{Z}^N\right\rbrace$, where the full-ranked matrix $\mathbf{G}\in\mathbb{R}^{N\times N}$ represents the generator matrix of $\Lambda$.
For any $\mathbf{x}\in\mathbb{R}^N$, finding its closest vector in a lattice $\Lambda$
is referred to as the closest vector problem (CVP). Based on efficient algorithms to solve CVP \cite{zamir2014lattice}, a nearest neighbor quantizer $Q_\Lambda(\cdot)$ is defined as
\begin{align}
	Q_\Lambda(\mathbf{x})=\mathop{\arg\min}_{\mathbf{\lambda} \in \Lambda} \|\mathbf{x}-\mathbf{\lambda}\|.
\end{align}
Details of some optimal low-dimensional lattices for quantization are listed in Table \ref{Tab. Appendix A}. Based on the CVP quantizer we also define $\rm{mod}(\mathbf{x},\Lambda) = \mathbf{x}-Q_{\Lambda}(\mathbf{x})$.

The set of vectors in $\mathbb{R}^{N}$ that are closer to $\lambda\in\Lambda$ than other vectors in $\Lambda$  is called the Voronoi region $\mathcal{V}_\lambda$, i.e.,
\begin{align}
	\mathcal{V}_\lambda=\{\mathbf{x}:Q_\Lambda(\mathbf{x})=\lambda\}.
\end{align}
The fundamental Voronoi region is represented as
\begin{align}
	\mathcal{V}_\Lambda=\{\mathbf{x}:Q_\Lambda(\mathbf{x})=\mathbf{0}\},
\end{align}
and the volume of the fundamental region is 
\begin{align}
	\rm{Vol}(\mathcal{V}_\Lambda)=\int_{\mathcal{V}_{\Lambda}} d\mathbf{x}=|\rm{det} \mathbf{G}|.
\end{align}
The packing  radius  $r_{\rm{pack}\left(\Lambda\right)}$, represents the radius of the largest sphere contained in $\mathcal{V}_\Lambda$:
\begin{align}
	r_{\rm{pack}(\Lambda)} = \frac{1}{2} \rm{min}_{\lambda \in \Lambda \backslash \{\mathbf{0}\}}\|\lambda\|.
\end{align}
The covering radius $r_{\rm{cov}\left(\Lambda\right)}$ refers to the radius of the smallest sphere containing $\mathcal{V}_\Lambda$:
\begin{equation}
	r_{\rm{cov}\left(\Lambda\right)}=\mathop{\max}_{\mathbf{x} \in \mathbb{R}^N} \mathop{\min}_{\lambda \in \Lambda} \|\mathbf{x}-\mathbf{\lambda}\|.
\end{equation}

In the lattice quantization, a typical and efficient method is using the nested lattices based shaping, called coset coding \cite{zamir2014lattice,moulin2005data}. Two lattices $\Lambda_f$ and $\Lambda_c$ are regarded as nested when these two lattices have a inclusion relation, i.e., $\Lambda_c\subset \Lambda_f$. The lattice $\Lambda_f$ is called the \emph{fine/coding} lattice, and its subset  $\Lambda_c$ is called  the \emph{coarse/shaping} lattice. The generator matrices of $\Lambda_c$ and $\Lambda_f$ which are expressed by $\mathbf{G}_c$ and $\mathbf{G}_f$, respectively, have a relation 
\begin{align}
	\mathbf{G}_c = \mathbf{G}_f \mathbf{J},
\end{align}
where the sub-sampling matrix $\mathbf{J}$ represents \emph{nesting matrix}. 
When $\mathbf{J}$ is an identity matrix, it denotes a self-similar shaping \cite{Kurkoski2018SelfSimilar} operation. When $\mathbf{J}$ is a diagonal matrix with different elements, it denotes a rectangular shaping \cite{Lyu2021RectangularShaping} operation.

\newsavebox{\boxa}
\savebox{\boxa}{$\mathbf{J}=\left[\begin{matrix}4&0\\0&4\end{matrix}\right]$}
\newsavebox{\boxb}
\savebox{\boxb}{$\mathbf{J}=\left[\begin{matrix}2&0\\0&4\end{matrix}\right]$}

\begin{table*}[t!]
	\centering
	\renewcommand\arraystretch{2}
	\caption{The generator matrices of some optimal low-dimensional quantizers \cite{conway2013sphere}.}  
	\label{Tab. Appendix A}  
	\begin{tabular}{c c c c c c }
		\hline
		
		\textbf{N}& Quantizer& Generator Matrix& $r_{pack}(\Lambda)$&$r_{cov}(\Lambda)$& $G(\Lambda)$ \\
		\hline
		
		1& $\mathbb{Z}$& 1
		& $\frac{1}{2}$& $\frac{1}{2}$
		&0.083333 \\
		
		2 & $A_2$& $\left[\begin{matrix}\frac{\sqrt{3}}{2}&0\\\frac{1}{2}&1\end{matrix}\right]$
		& $\frac{1}{2}$& $\frac{\sqrt{3}}{3}$
		& 0.080188 \\
		
		3 & $A_3 \cong D_3$& $\left[\begin{matrix} -1&1&0\\-1&-1&1\\0&0&-1\end{matrix}\right]$
		& $\frac{\sqrt{2}}{2}$& $1$
		& 0.078745 \\
		
		4 & $D_4$& $\left[\begin{matrix} 2&1&1&1\\0&1&0&0\\0&0&1&0\\0&0&0&1\end{matrix}\right]$
		& $\frac{\sqrt{2}}{2}$& $1$
		& 0.076603 \\
		
		8 & $E_8$
		& $\left[\begin{matrix} 
		2&-1&0&0&0&0&0&\frac{1}{2}\\
		0&1&-1&0&0&0&0&\frac{1}{2}\\
		0&0&1&-1&0&0&0&\frac{1}{2}\\
		0&0&0&1&-1&0&0&\frac{1}{2}\\
		0&0&0&0&1&-1&0&\frac{1}{2}\\
		0&0&0&0&0&1&-1&\frac{1}{2}\\
		0&0&0&0&0&0&1&\frac{1}{2}\\
		0&0&0&0&0&0&0&\frac{1}{2}\\
		\end{matrix}\right]$
		& $\frac{\sqrt{2}}{2}$& $1$
		& 0.071682 \\
		\hline
	\end{tabular}
\end{table*}

\section{The Proposed Method}\label{Sec. Mathmetical Model}

\subsection{Lattice-QIM reformulation}\label{Sec.QIM improve}
Following \cite{moulin2005data,9455352,conway2013sphere},  we reformulate   QIM for arbitrary lattices and payloads from the perspective of lattice quantization and coset coding, based on which the description of the proposed RAW scheme can also be simpler.

\subsubsection{Embedding}

\noindent \textbf{Step i)} Choose a pair of nested lattices $\Lambda_f$ and $\Lambda_c \subset \Lambda_f$ in $\mathbb{R}^N$. 
The fine lattice $\Lambda_f$ can be decomposed as the union of $|\rm{det} \mathbf{J}|$ cosets of the coarse lattice $\Lambda_c$:
\begin{align}\label{Eq. lambda f&c}
	\Lambda_f = \bigcup\limits_{i=0}^{|\rm{det} \mathbf{J}|-1} \Lambda_i
	=\bigcup\limits_{\mathbf{d}_i \in \Lambda_f \backslash \Lambda_c} (\mathbf{d}_i + \Lambda_c)
\end{align}
where each coset $\Lambda_i = \Lambda_c  + \mathbf{d}_i$ is a translated coarse lattice and $\mathbf{d}_i$ is called the coset representative of $\Lambda_i$.

\noindent  \textbf{Step ii)} Perform labeling to associate the $i\in \left\lbrace 0, 1, \ldots, |\rm{det} \mathbf{J}|-1 \right\rbrace $  to the $i$th message vector $\mathbf{m}_i \in \mathcal{M}$.

\noindent \textbf{Step iii)}  
Consider a cover/host vector denoted as $\mathbf{s}$. To embed a message vector $\mathbf{m}_i$ over $\mathbf{s}$, the watermarked cover vector is set as
\begin{equation} \label{Eq. QIM_encode}
	\mathbf{s}_{\mathrm{QIM}} \triangleq	Q_{\Lambda_i}(\mathbf{s}) = Q_{\Lambda_c}(\mathbf{s}-\mathbf{d}_i)+\mathbf{d}_i.
\end{equation}

\subsubsection{Decoding}
Let the (possibly) noisy observation of $	\mathbf{s}_{\mathrm{QIM}}$ be $\mathbf{y}$.

\noindent \textbf{Step i)}   Estimate the coset representative $\mathbf{d}_{i}$ of $\Lambda_i$ by
\begin{equation}\label{Eq.qim_decode}
	\hat{\mathbf{d}}_{i}={\rm mod}(Q_{\Lambda_f}(\mathbf{y}),\Lambda_c).
\end{equation}

\noindent \textbf{Step ii)} Perform inverse labeling to turn $\hat{\mathbf{d}}_{i}$ into the estimated message vector   $\hat{\mathbf{m}}_i$.

The extracted message vector is correct as long as $Q_{\Lambda_f}(\mathbf{y})=\mathbf{s}_{\mathrm{MME}}$. Moreover, the information transmission rate per dimension of Lattice-QIM is simply 
\begin{equation} 
	R=\frac{1}{N}\log {|\rm{det} \mathbf{J}|}.
\end{equation}

\subsection{Meet-in-the-middle embedding}\label{SSec. GLMR}

Based on the above, we observe that there exists a difference vector  $\mathbf{e}$ between the cover vector $\mathbf{s}$ and its quantized watermarked vector $	\mathbf{s}_{\mathrm{QIM}}$, i.e., 
\begin{align}
	\mathbf{e} =  \mathbf{s} - Q_{\Lambda_i}(\mathbf{s}).
\end{align}
Obviously the information about $\mathbf{e}$ is lost if we only use $Q_{\Lambda_i}(\mathbf{s})$ as the watermarked vector.

Notice that QIM has certain error tolerance capability. If we treat $\mathbf{e}$ as the ``beneficial noise'' and add it back to $Q_{\Lambda_i}(\mathbf{s})$, then the information about the cover $\mathbf{s}$ can be maintained, and the scheme becomes reversible. The tricky part about adding  ``beneficial noise'' is that, they should be properly scaled to meet several demands.  First, the scaled $\mathbf{e}$  should be small enough such that it does not go beyond the Voronoi region of the fine lattice. Second, the scaled $\mathbf{e}$  should not be too small to avoid exceeding the used representation accuracy of numbers. Details of the proposed scheme is explained as follows.

\begin{definition}[MME Embedding]\label{Def.reversible lattice model definition}
	The watermarked cover vector is set as
	\begin{align}\label{Eq. REF}
		\mathbf{s}_{\mathrm{MME}} = \mathcal{W}(\mathbf{s}) &= Q_{\Lambda_i}\left(\mathbf{s}\right)+\left(1-\alpha\right)(\mathbf{s}-	Q_{\Lambda_i}\left(\mathbf{s}\right)) \notag\\
		&=  \alpha Q_{\Lambda_i}\left(\mathbf{s}\right) + \left(1-\alpha\right) \mathbf{s},
	\end{align}
	in which $\alpha$ is a chosen scaling factor such that 
	$(1-\alpha)(\mathbf{s}-Q_{\Lambda_i}\left(\mathbf{s}\right))\in\mathcal{V}_{\Lambda_f}$.
\end{definition}

\begin{definition}[Noiseless MME Extraction]
	If the receiver's side has $\mathbf{s}_{\mathrm{MME}}$,
	the estimated cover vector can be accurately restored by
	\begin{align}\label{Eq. RF}
		\hat{\mathbf{s}}
		=\mathcal{W}^{-1}\left(\mathbf{s}_{\mathrm{MME}}\right)
		&=\frac{\mathbf{s}_{\mathrm{MME}} -
			Q_{\Lambda_f}\left(\mathbf{s}_{\mathrm{MME}}\right)
		}{1-\alpha} + Q_{\Lambda_f}\left(\mathbf{s}_{\mathrm{MME}}\right)
		\nonumber \\
		&= \frac{1}{1-\alpha}\mathbf{s}_{\mathrm{MME}}-\frac{\alpha}{1-\alpha} Q_{\Lambda_f}\left(\mathbf{s}_{\mathrm{MME}}\right).
	\end{align}
	In this case, the message vector can also be accurately decoded by using the decoder of QIM. 
\end{definition}

In the noisy setting,
consider an additive white Gaussian noise (AWGN) channel in the form of $\mathbf{y}=\mathbf{s}_{\mathrm{MME}}+\mathbf{n}$, where the entries of $\mathbf{n}$ admit a standard Gaussian distribution $\mathcal{N}(0,\sigma_{\mathbf{n}}^2)$. In this case, accurately restring the cover vector is impossible, but the message vector may be accurately decoded by using the decoder of QIM. 
Compared with IQIM, MME is advantageous in the following aspects.
\begin{itemize}
	\item As shown in  Figs. \ref{Fig. embeddin process comparison one dimension} and  \ref{Fig. embeddin process comparison two dimension}, the difference vector of MME can be confined to a small Voronoi cell (hence the name ``meet-in-the-middle'') while IQIM cannot. In the distortion analysis we will shown that this brings certain performance gains. 
	\item
	MME can be applied to general lattices while IQIM is only feasible for the standard integer lattice.
\end{itemize}



\begin{figure}[t!]
	\linespread{1}
	\centering
	\subfigure[]{\includegraphics[width=.45\textwidth]{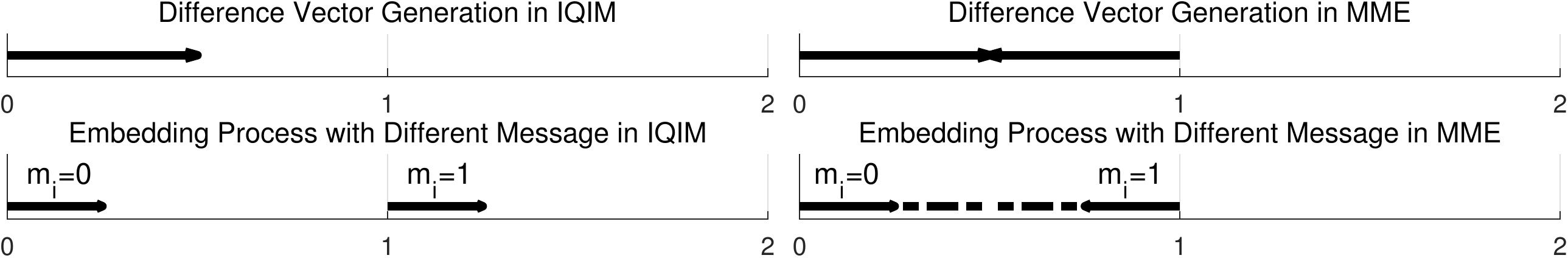}}
	\quad 
	\subfigure[]{\includegraphics[width=.45\textwidth]{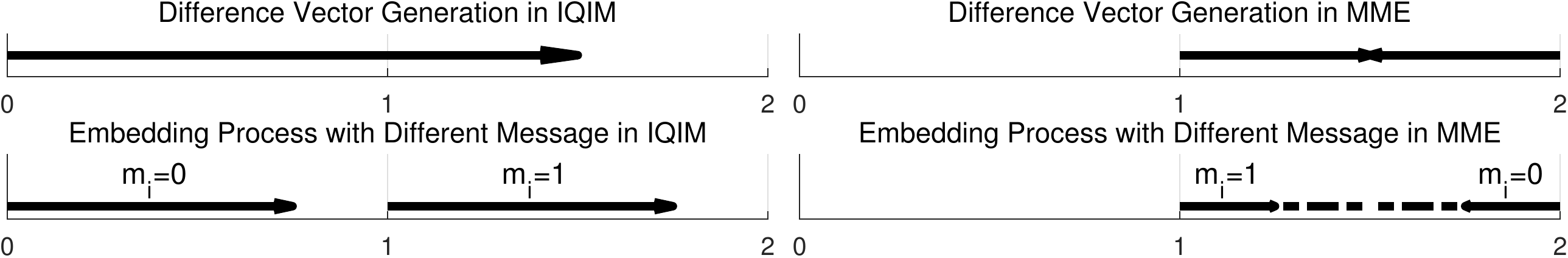}} 
	
	\caption{Comparison of embedding process between IQIM and MME in one-dimension with different host vector $\mathbf{s}$. (a) $\mathbf{s}=0.5$. (b) $\mathbf{s}=1.5$.} 
	\label{Fig. embeddin process comparison one dimension}
\end{figure}
\begin{figure}[t!]
	\linespread{1}
	\centering
	\includegraphics[width=.45\textwidth]{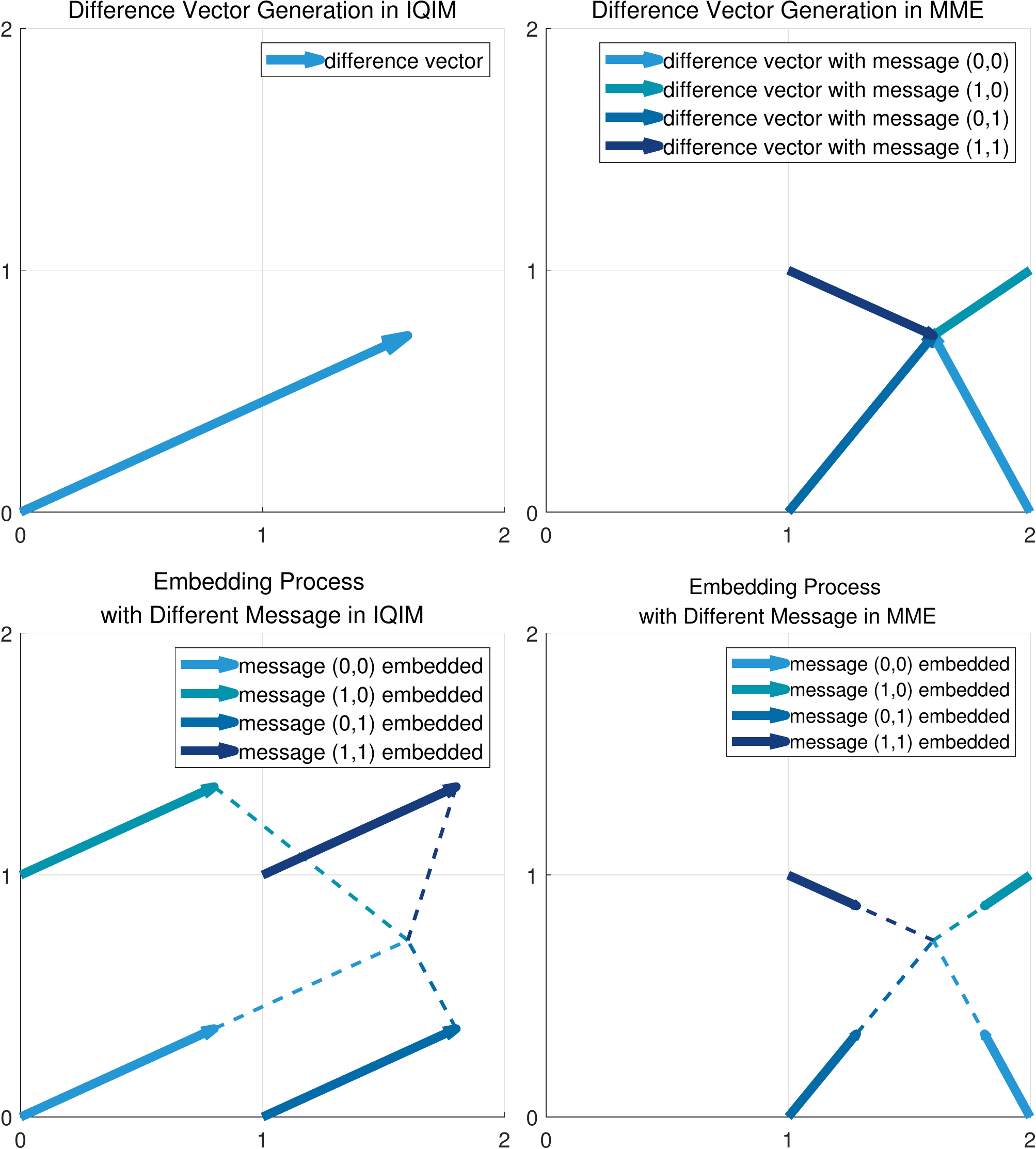}
	
	\caption{Comparison of embedding process between IQIM and MME in two-dimension with $\mathbf{s}=(1.61,0.73)$.} 
	\label{Fig. embeddin process comparison two dimension}
\end{figure}

\subsection{Correctness analysis}
\subsubsection{Reversibility}\label{SSec. Reversibility of GLMR}
For the embedding function Eq. (\ref{Eq. REF}), its definition field $\mathcal{A}$ is noted as
\begin{equation}
	\mathcal{A}=\{\mathbf{s}|\mathbf{s}\in\mathbb{R}^N\}.
\end{equation}
And the value field $\mathcal{B}$ is defined as
\begin{equation}\label{defsujection}
	\mathcal{B}=\{\mathbf{s}_{\mathrm{MME}}|\mathbf{s}_{\mathrm{MME}}=\mathcal{W}\left(\mathbf{s}\right),\mathbf{s}\in\mathcal{A}\}.
\end{equation}
Hereby we introduce a theorem about the one-to-once correspondence between about the cover and the watermarked cover.

\begin{theorem}[Noiseless Recovery]
	\label{Lem. lemma one to one}
	The function $\mathcal{W}$ in MME is bijective.\end{theorem}
\begin{proof}
	It suffices to prove that $\mathcal{W}$ is both injective and  surjective.
	``Injective'' means no two elements in the domain of the function gets mapped to the same image, i.e.,
	for any cover vectors $\mathbf{s}_a,\mathbf{s}_b\in\mathcal{A}$,
	\begin{equation}
		\mathbf{s}_a\neq\mathbf{s}_b\rightarrow \mathcal{W}\left(\mathbf{s}_a\right)\neq\mathcal{W}\left(\mathbf{s}_b\right).
	\end{equation}	
	With
	$\mathbf{s}_a\neq\mathbf{s}_b$, we analyze the relationship between their output $\mathcal{W}\left(\mathbf{s}_a\right)$ and $\mathcal{W}\left(\mathbf{s}_b\right)$ in three cases:
	\begin{itemize}
		\item[1.]$\mathbf{s}_a$ and $\mathbf{s}_b$ belong to two different Voronoi regions $\mathcal{V}_{\lambda_a}$ and $\mathcal{V}_{\lambda_b}$ of the coarse lattice $\Lambda_c$.  It follows from  $\mathcal{V}_{\lambda_a}\cap\mathcal{V}_{\lambda_b}=\emptyset$ that
		$\mathcal{W}\left(\mathbf{s}_a\right)\neq\mathcal{W}\left(\mathbf{s}_b\right)$.
		\item[2.]   $\mathbf{s}_a$ and $\mathbf{s}_b$ belong to the same Voronoi region in $\Lambda_c$ but different watermarks ($Q_{\Lambda_i}\left(\mathbf{s}_a\right)\neq Q_{\Lambda_i}\left(\mathbf{s}_b\right)$). Since the ``self-noise'' does not go beyond the Voronoi cell of a fine lattice, i.e., $(1-\alpha)(\mathbf{s}-Q_{\Lambda_i}\left(\mathbf{s}\right))\in\mathcal{V}_{\Lambda_f}$, then obviously  $\mathcal{W}\left(\mathbf{s}_a\right)\neq\mathcal{W}\left(\mathbf{s}_b\right)$.
		\item[3.]   $\mathbf{s}_a$ and $\mathbf{s}_b$ belong to the same Voronoi region in $\Lambda_c$ and watermark ($Q_{\Lambda_i}\left(\mathbf{s}_a\right)= Q_{\Lambda_i}\left(\mathbf{s}_b\right)$). Since $\mathbf{s}_a\neq\mathbf{s}_b$, one has $(1-\alpha)\mathbf{s}_a\neq(1-\alpha)\mathbf{s}_b$, and they are still unequal after both adding $Q_{\Lambda_i}\left(\mathbf{s}\right)$.
	\end{itemize}
	Combining the above three cases, the injection is proved.
	
	``Surjective'' means that any element in the range of the function is hit by the function.  We prove this by using contradiction.	If $\mathcal{W}$ is  not surjective, i.e., $\exists \mathbf{s}_{\mathrm{MME}}\in\mathcal{B}$, $\forall \mathbf{s}\in\mathcal{A}$ \textit{s.t.} $\mathcal{W}(\mathbf{s})\neq\mathbf{s}_{\mathrm{MME}}$, then we obtain $\mathbf{s}_{\mathrm{MME}}\notin\mathcal{B}$ which   contradicts Eq. (\ref{defsujection}). 
\end{proof}
The above theorem says that $\mathcal{W}$ is a
reversible function, where each element of one set is paired with exactly one element of the other set. Thus the correctness of the estimated covers is guaranteed.

In the noisy scenario of $\mathbf{n}\neq \mathbf{0}$ in $\mathbf{y}=\mathbf{s}_{\mathrm{MME}}+\mathbf{n}$, we argue that the cover vector is also approximately reversible if  $
Q_{\Lambda_f}\left(\hat{\mathbf{v}}\right)=\mathbf{0}$,
where 
\begin{equation}\label{Eq. whole noise2}
	\hat{\mathbf{v}} =  (1-\alpha)\mathbf{e}
	+ \mathbf{n}
\end{equation} 
is referred to as the composite noise vector at the receiver's side. 
To be concise, by using the estimation function
\begin{align}
	\hat{\mathbf{s}}_{\mathrm{noisy}}
	=\mathcal{W}^{-1}\left(\mathbf{y}\right)
	&=\frac{\mathbf{y} -
		Q_{\Lambda_f}\left(\mathbf{y}\right)
	}{1-\alpha} + Q_{\Lambda_f}\left(\mathbf{y}\right),
\end{align}
we have
\begin{equation}
	\hat{\mathbf{s}}_{\mathrm{noisy}}-\mathbf{s}=\frac{1}{1-\alpha}\mathbf{n}.
\end{equation}
It says that, if the composite noise $\hat{\mathbf{v}}$ is small enough to ensure the correctness of the message, then the distance from the noisy estimated cover and the clean cover is only $||{1}/{(1-\alpha)}\mathbf{n}||$.

\subsubsection{Correctness of the estimated messages}\label{SSec. RF}
We have the following theorem about the robustness of the watermarking scheme. Since the ``self-noise'' is within the boundary of the Voronoi region, the estimated messages in the noiseless setting is correct for sure. More generally in the noisy setting, if the composite noise is small, the correctness of the estimated messages is also guaranteed.

\begin{theorem}\label{Lem. lemma recover function 1}
	If $\hat{\mathbf{v}}$ satisfies
	\begin{equation}\label{Eq. noise condition}
		Q_{\Lambda_f}(\hat{\mathbf{v}})=\mathbf{0},
	\end{equation} 
	then the estimated message by using MME is correct.
\end{theorem}
\begin{proof}
	At the receiver's side, the noisy observation can be written as   $\mathbf{y}=Q_{\Lambda_i}\left(\mathbf{s}\right)+\hat{\mathbf{v}}$.
	Since $Q_{\Lambda_i}\left(\mathbf{s}\right)\in\Lambda_f=\cup_{\mathbf{d}_i \in \Lambda_f \backslash \Lambda_c} (\mathbf{d}_i + \Lambda_c)$, $Q_{\Lambda_f}(\mathbf{y})$ can be written as
	\begin{equation}
		Q_{\Lambda_f}(\mathbf{y})=Q_{\Lambda_f}(Q_{\Lambda_i}\left(\mathbf{s}\right))+Q_{\Lambda_f}(\hat{\mathbf{v}}).
	\end{equation}
	Then we have
	\begin{align}\label{Eq. decode proof}
		&{\rm mod}(Q_{\Lambda_f}(\mathbf{y}),\Lambda_c)\notag\\
		=&{\rm mod}(Q_{\Lambda_f}(Q_{\Lambda_i}\left(\mathbf{s}\right))
		+Q_{\Lambda_f}(\hat{\mathbf{v}}),\Lambda_c)\notag\\
		=&{\rm mod}({\rm mod}(Q_{\Lambda_f}(Q_{\Lambda_i}\left(\mathbf{s}\right)),\Lambda_c)
		+{\rm mod}(Q_{\Lambda_f}(\hat{\mathbf{v}}),\Lambda_c),\Lambda_c)\notag\\
		=&{\rm mod}(\mathbf{d}_i+Q_{\Lambda_f}(\hat{\mathbf{v}}),\Lambda_c).
	\end{align}
	Thus ${\rm mod}(Q_{\Lambda_f}(\mathbf{y}),\Lambda_c)=\mathbf{d}_i$ if $Q_{\Lambda_f}(\hat{\mathbf{v}})=\mathbf{0}$. Since each coset representative $\mathbf{d}_i$ corresponds a unique message $\mathbf{m}_i$,  the lemma is proved.
\end{proof}


\subsection{Setting the scaling factor} \label{SSec. proposed lower bound}
In this subsection, we will discuss the feasible range of $\alpha$. 
\subsubsection{Lower bound} 
The lower bound of $\alpha$ is determined by ensuring the correctness of estimated messages.  As mentioned in Definition 1, one should have  
$(1-\alpha)(\mathbf{s}-Q_{\Lambda_i}\left(\mathbf{s}\right))\in\mathcal{V}_{\Lambda_f}$. This implies that
\begin{equation}\label{Eq. radius inequality}
	(1-\alpha)\mathcal{V}_{\Lambda_c}\subseteq\mathcal{V}_{\Lambda_f}.
\end{equation}
It suffices to have
\begin{equation}
	(1-\alpha)r_{cov}(\Lambda_c)\leq r_{pack}(\Lambda_f),
\end{equation}
which means 
\begin{equation}\label{alphabound1}
	\alpha\geq 1-\frac{r_{pack}(\Lambda_f)}{r_{cov}(\Lambda_c)}.
\end{equation}

For the special case of self similar shaping, a tighter bound can be derived. To be concise, let $\Lambda_c = \Gamma \Lambda_f$.  It follows from $	(1-\alpha) \Gamma \mathcal{V}_{\Lambda_f}\subseteq\mathcal{V}_{\Lambda_f}$ that
\begin{equation}\label{alphabound2}
	\alpha \geq 1- 1/\Gamma.
\end{equation}

\newsavebox{\boxxa}
\savebox{\boxxa}{$\mathbf{J}=\left[\begin{matrix}4&0\\0&4\end{matrix}\right]$}
\newsavebox{\boxxb}
\savebox{\boxxb}{$\mathbf{J}=\left[\begin{matrix}4&0\\0&4\end{matrix}\right]$}
\newsavebox{\boxxc}
\savebox{\boxxc}{$\mathbf{J}=\left[\begin{matrix}10&0\\0&11\end{matrix}\right]$}
\newsavebox{\boxxd}
\savebox{\boxxd}{$\mathbf{J}=\left[\begin{matrix}2&0\\0&4\end{matrix}\right]$}
\newsavebox{\boxxe}
\savebox{\boxxe}{$\mathbf{J}=\left[\begin{matrix}2&0\\0&4\end{matrix}\right]$}

\begin{figure}[t!]
	\linespread{1}
	\centering
	\subfigure[]{\includegraphics[width=.22\textwidth]{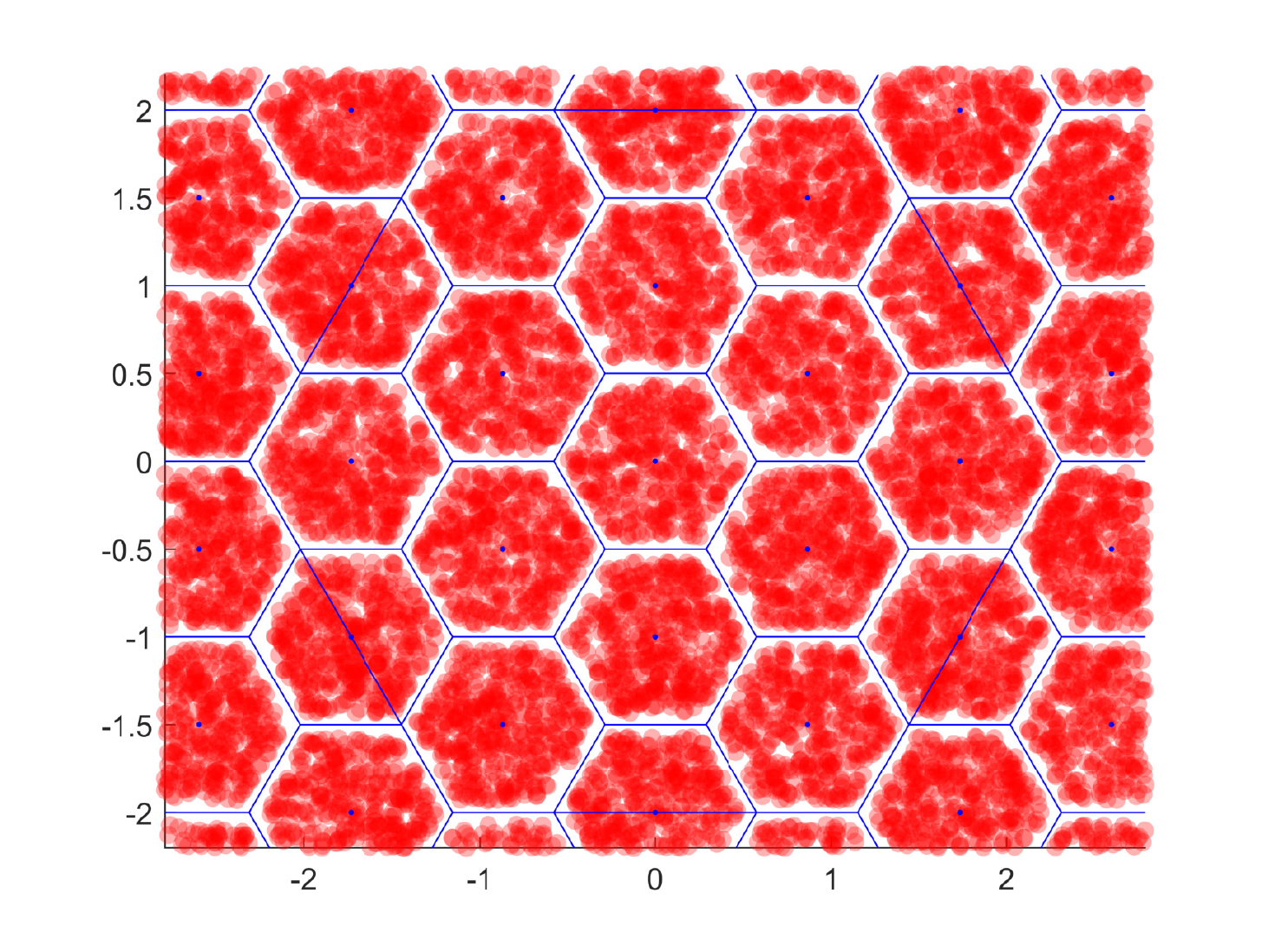}}
	\subfigure[]{\includegraphics[width=.22\textwidth]{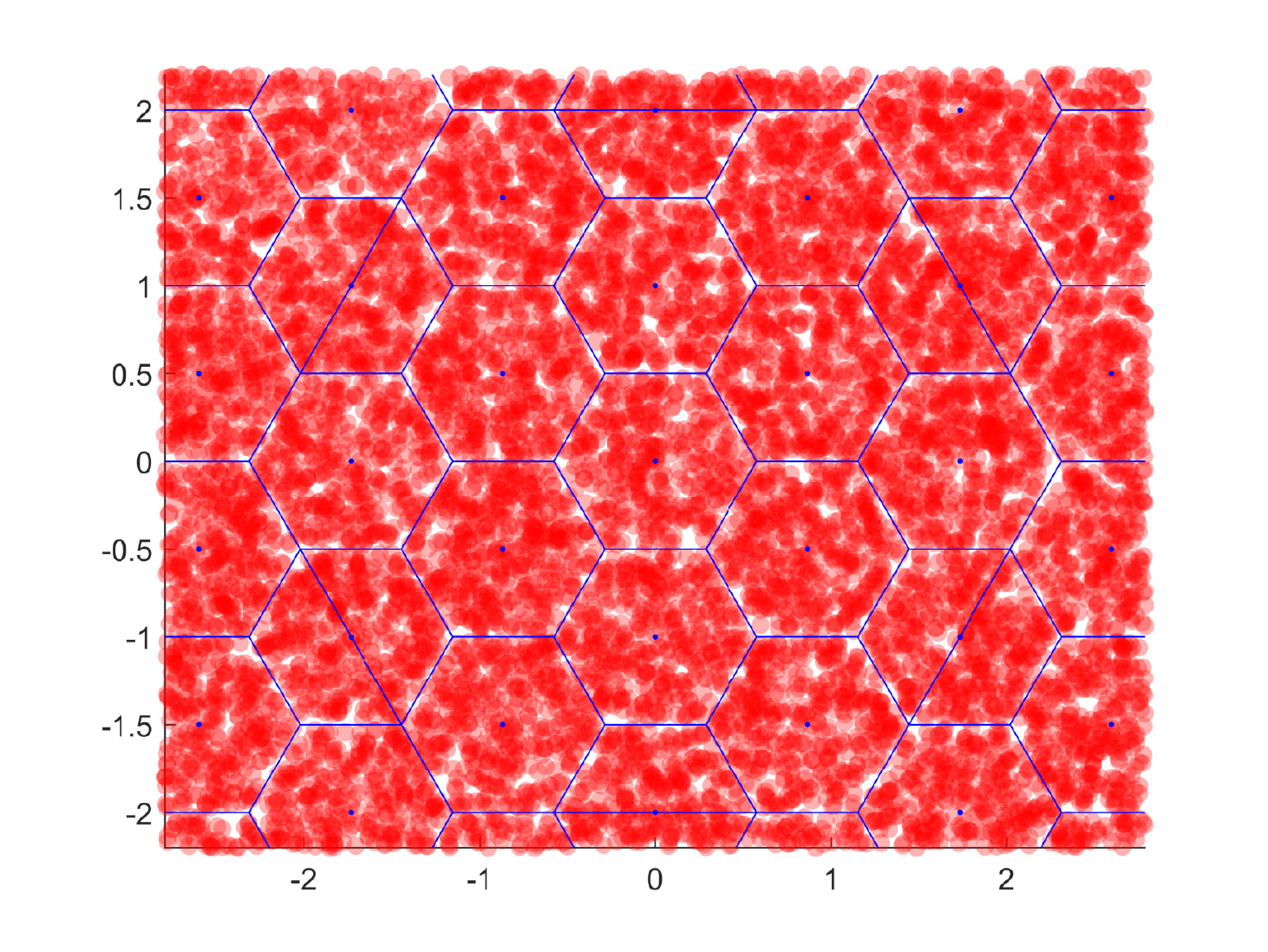}}
	\caption{Distribution of the watermarked cover vector $\mathbf{s}_{\mathrm{MME}}$ with $\Lambda_f$ being the $A_2$ lattice and $\Lambda_c = 4 \Lambda_f$. 
		(a)  $\alpha = 1 - \frac{\sqrt{3}}{6}$ based on Eq. (\ref{alphabound1}). 
		(b)  $\alpha = 3/4$ based on Eq. (\ref{alphabound2}).} 
	\label{Fig. Cases of overstep lower bound}
\end{figure}

Fig. \ref{Fig. Cases of overstep lower bound} depicts the distribution of the watermarked cover vector $\mathbf{s}_{\mathrm{MME}}$ based on the general bound in Eq. (\ref{alphabound1}) and the special-case bound in Eq. (\ref{alphabound2}).

\subsubsection{Upper bound}
The upper bound of $\alpha$ is related to the correct extraction of the covers.  Recall that
the audio data is stored as single or double precision floating-point numbers (as shown in Table \ref{Tab. Appendix B}), and it cannot represent arbitrary real-valued numbers without loss of accuracy. Overflow occurs when the magnitude of a number exceeds the range allowed by the size of the bit field.
To avoid the effect of overflow, we need to have 
\begin{equation}\label{alphaupperbound}
	\alpha\leq 1-2^{-(\mathcal{L}+1)},
\end{equation}
where  $\mathcal{L}$ denotes the mantissa bit length. For example, we have $\alpha \leq 1-2^{-53}$ when using double precision floating-point numbers.

\begin{table*}[t!]
	\centering
	\renewcommand\arraystretch{2}
	\caption{The storage structure of floating-point numbers.}  
	\label{Tab. Appendix B}  
	\begin{tabular}{c c c c c }
		\hline
		
		Data Type& Total length (bit)& Sign length (bit)& Order length (bit)& Mantissa length (bit)\\
		\hline
		Float & 32& 1& 8& 23\\
		Double & 64& 1& 11& 52\\
		\hline
	\end{tabular}
\end{table*}

\section{Performance metrics}\label{SSec. Theoretical Analysis}


\subsection{Embedding distortion}
To evaluate the embedding distortion, we adopt the signal to watermark ratio (SWR) metric defined as
\begin{align}\label{Eq. SWR1}
	{SWR}~(\mathrm{dB})=10\times lg\left(\frac{\sigma^2_\mathbf{s}}{\sigma^2_\mathbf{w}}\right),
\end{align}
where $\sigma^2_\mathbf{s}$, $\sigma^2_\mathbf{w}$ represent the power of the host and the additive watermark, respectively. For MME, we have
\begin{align}
	\sigma^2_\mathbf{w} 	&=   E(\|\alpha \mathbf{e}\|^2)\notag\\
	&=\alpha^2 E(\| \mathbf{e}\|^2)\notag \\
	&= \alpha^2 N G({\Lambda_c}) \rm{Vol}(\mathcal{V}_{\Lambda_c})^{2/N}, \label{eq_sigma_w2}
\end{align} 
in which $G(\Lambda_c)$ represents the normalized second moment of the chosen coarse lattice $\Lambda_c$, and Eq. (\ref{eq_sigma_w2}) is derived from the widely adopted flat-host assumption \cite{DBLP:journals/tcsv/WangZLCH21}. For
the simplest setting of $\Lambda_f = \Delta \mathbb{Z}^N$ and $\Lambda_c = \Delta 2^b \mathbb{Z}^N$, since $G(\Delta 2^b \mathbb{Z}^N)=1/12$, we have
\begin{equation}\label{cubic_dist}
	\sigma_{\mathbf{w}}^2
	=\frac{N\alpha^2 2^{2b}\Delta^2}{12}.
\end{equation}


\begin{proposition}\label{The. greater SWR}
	By setting $\Lambda_f = \Delta \mathbb{Z}^N$, $\Lambda_c = \Delta 2^b \mathbb{Z}^N$, and $\alpha = \beta$, 
	MME has a larger SWR than IQIM.
\end{proposition}
\begin{proof}	
	Based on the inequality that
	\begin{equation}
		x-1<\left\lfloor x\right\rfloor\leq x,
	\end{equation}
	the embedding distortion of IQIM can be bounded as 
	\begin{align}
		&	\sigma_{\mathbf{w},\rm{IQIM}}^2 \notag\\
		&=E\left(\beta^2 \left\lVert \mathbf{s} - \left[2^b\Delta \left[\gamma_1~\cdots~\gamma_N\right]^\top+\frac{\Delta\mathbf{m}}{\beta}\right] \right\rVert^2 \right)\notag\\
		&=\frac{\beta^2}{2^b}\sum_{q=1}^{N}\sum_{m_{i,q}=0}^{2^b-1}\int_{0}^{2^b\Delta}\left[2^b\Delta\gamma_q+\frac{m_{i,q}\Delta}{\beta}-s_q\right]^2\notag\\
		&\cdot f(s_q) ds_q\notag\\
		&>\frac{\beta^2}{2^b}\sum_{q=1}^{N}\sum_{m_{i,q}=0}^{2^b-1}\int_{0}^{2^b\Delta}\left[\frac{m_{i,q}\Delta}{\beta}-2^b\Delta\right]^2 f(s_q) ds_q \notag\\
		&=\frac{\left(2^b\right)^2\left(2\cdot 2^b-1\right)}{6\left(2^b-1\right)}N\beta^2\Delta^2.
	\end{align}
	Letting $\alpha = \beta$ in Eq. (\ref{cubic_dist}), we have
	\begin{equation}
		\sigma_{\mathbf{w},\rm{IQIM}}^2 - 	\sigma_{\mathbf{w}}^2 > \frac{\beta^2 2^{2b} \left(3\cdot 2^b-1\right)N\Delta^2}{12\left(2^b-1\right)}>0. 
	\end{equation}
	As IQIM has a larger denominator in SWR, the proposition is proved.
\end{proof}
Proposition \ref{The. greater SWR} justifies the advantage of MME over IQIM: by using a 
better meet-in-the-middle approach to construct the ``beneficial noise'', it enjoys a larger SWR.

\subsection{Robustness of messages}\label{SSec. gsnr}
In addition to the distortion/imperceptibility metric, it is also necessary to evaluate the robustness of MME against additive noises.
Following \cite{moulin2005data}, we define the 
generalized signal-to-noise ratio (GSNR) as
\begin{align}
	{\rm GSNR}&\triangleq\frac{\min_{i,j\in\mathcal{I},i\neq j}\|\mathbf{d}_i-\mathbf{d}_j\|^2}{ E(\| \hat{\mathbf{v}}\|^2) }\notag\\
	&=\frac{4r_{{\rm{pack}}(\Lambda_f)}^2}{E(\| \hat{\mathbf{v}}\|^2)}, \label{eqgsnr_std}
\end{align}
where $\hat{\mathbf{v}}$ is the composite noise vector defined in Eq. (\ref{Eq. whole noise2}). 

Due to the independence of $\mathbf{e}$ and $\mathbf{n}$, we have 
\begin{align}
	E(\| \hat{\mathbf{v}}\|^2) &= 	E(\|(1-\alpha)\mathbf{e}
	+ \mathbf{n} \|^2) \notag \\
	&=(1-\alpha)^2 	E(\| {\mathbf{e}}\|^2) + E(\| {\mathbf{n}}\|^2) \notag \\
	&=(1-\alpha)^2 N G({\Lambda_c}) {\rm{Vol}}(\mathcal{V}_{\Lambda_c})^{2/N} + N \sigma_{\mathbf{n}}^2.\label{eq3GSNR}
\end{align}

By substituting (\ref{eq3GSNR}) into (\ref{eqgsnr_std}), the 
GSNR of MME can be expressed as
\begin{equation}
	{\rm GSNR} = \frac{4r_{{\rm{pack}}(\Lambda_f)}^2}{(1-\alpha)^2 N G({\Lambda_c}) {\rm{Vol}}(\mathcal{V}_{\Lambda_c})^{2/N} + N \sigma_{\mathbf{n}}^2}.
\end{equation}

Since the energy of the self-noise $(1-\alpha) \mathbf{e}$ is smaller than that of IQIM, its composite-noise also features a smaller energy. Therefore, MME has a larger GSNR than IQIM.

\section{Simulations}\label{Sec. equipment}
We carry out some simulations in this section to verify the effectiveness of MME. The simulation setups are summarized as follows.

\begin{table}[t!]
	\renewcommand{\arraystretch}{1}
	\centering  
	\caption{The level of objective difference grade scores \cite{xiang2017digital}.}  
	\label{Tab. objective difference grade}  
	\begin{tabular}{l l l}
		\hline
		ODG &Quality & Impairment\\
		\hline
		~0 & Excellent & Imperceptible\\				
		-1 & Good & Perceptible but not annoying\\				
		-2 & Fair & Slightly annoying\\				
		-3 & Poor & Annoying\\			
		-4 & Bad & Very annoying\\
		\hline
	\end{tabular}\label{ta_odg3}
\end{table} 

\noindent \textbf{Datasets:} In the experiments, two online datasets \cite{database1} and \cite{database2} are adopted, which contain ringtone, natural sound, absolute music, songs, and speeches. They consist of $5$ and $70$ bipolar $16$-bits Waveform Audio File Format (WAV) files, respectively, which are  labeled as No. 1001-1005 and 2001-2070. The cover vectors are generated from these datasets, and the 
messages are generated from uniform random distributions.

\noindent \textbf{Distortion metrics:}
In addition to the SWR metric, the objective difference grade (ODG) \cite{xiang2017digital} is also introduced in the experiments, which is a popular metric to assess accurately the objective auditory perception and is produced by the Perceptual Evaluation of Audio Quality (PEAQ) measure standardized in the ITU-R BS. 1387 \cite{itu20011387}.  The ODG scores are computed by the software downloaded from \cite{PEAQsoftware}, which is the Matlab version of Evaluation of Audio Quality (EAQUAL) \cite{lerch2002zplane}. ODG gives a score in the range of  $[-4.5,0.5]$. 
As shown in Table \ref{ta_odg3}, the audio quality is said to be excellent, good, fair, poor, and bad based on the respective ODG scores.

\noindent \textbf{Robustness metrics:}
In addition to the GSNR metric, we also introduce the bit error rate (BER) metric to exactly show the decoding error rates. 
BER is defined as
\begin{equation}
	BER=\frac{\sum_{k=1}^{M}\sum_{q=1}^{N}\left(\mathbf{x}_{k,q}\oplus \mathbf{x}'_{k,q}\right)}
	{M*N},
\end{equation}
where $\mathbf{x}$ and $\mathbf{x}'$ separately represent the original and estimated messages, while $q$ and $k$ are dimensional and sample identifiers respectively.

\noindent \textbf{Benchmarks and settings:} The state-of-the-art methods in \cite{peng2011effective,LIANG2020107584,nishimura2013reversible} are adopted as the benchmarks, which are respectively named as ``IQIM'', ``Liang'' and ``Nishimura''.  
For IQIM, Nishimura and the proposed MME, unless stated otherwise, we set  $\alpha=0.6569$, $\Delta=2000$, $\Lambda_f=\Delta\mathbb{Z}^2$, and $\Lambda_c=\Delta2^R\mathbb{Z}^2$.  
In addition, ``Liang'' is configured with parameters 
$S=300,~T=7500,~G=20000$.

\subsection{Imperceptibility performance}

\begin{figure}[t!]
	\linespread{1}
	\centering
	\subfigure[]{\includegraphics[width=.35\textwidth]{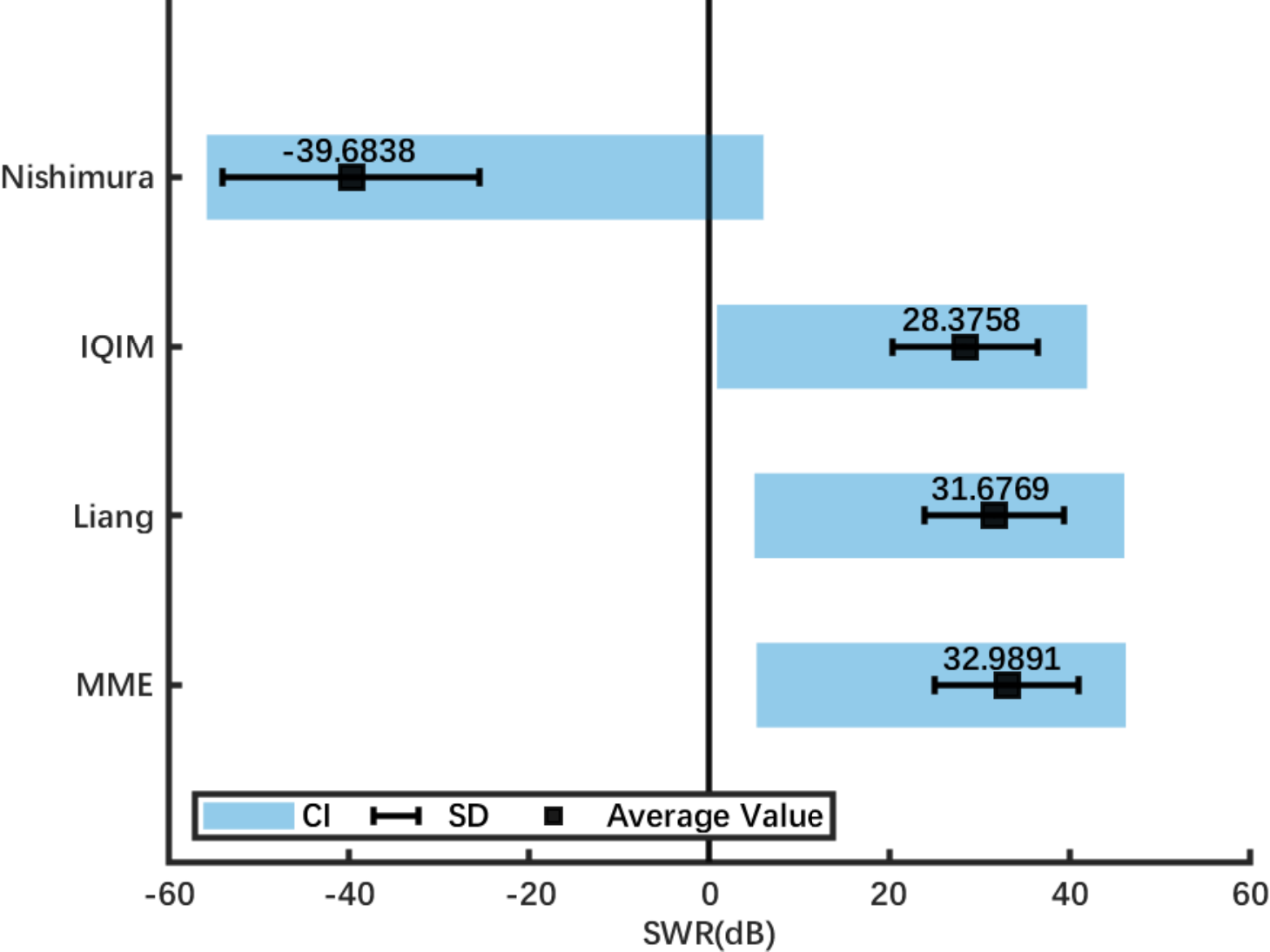}}
	\quad
	\subfigure[]{\includegraphics[width=.35\textwidth]{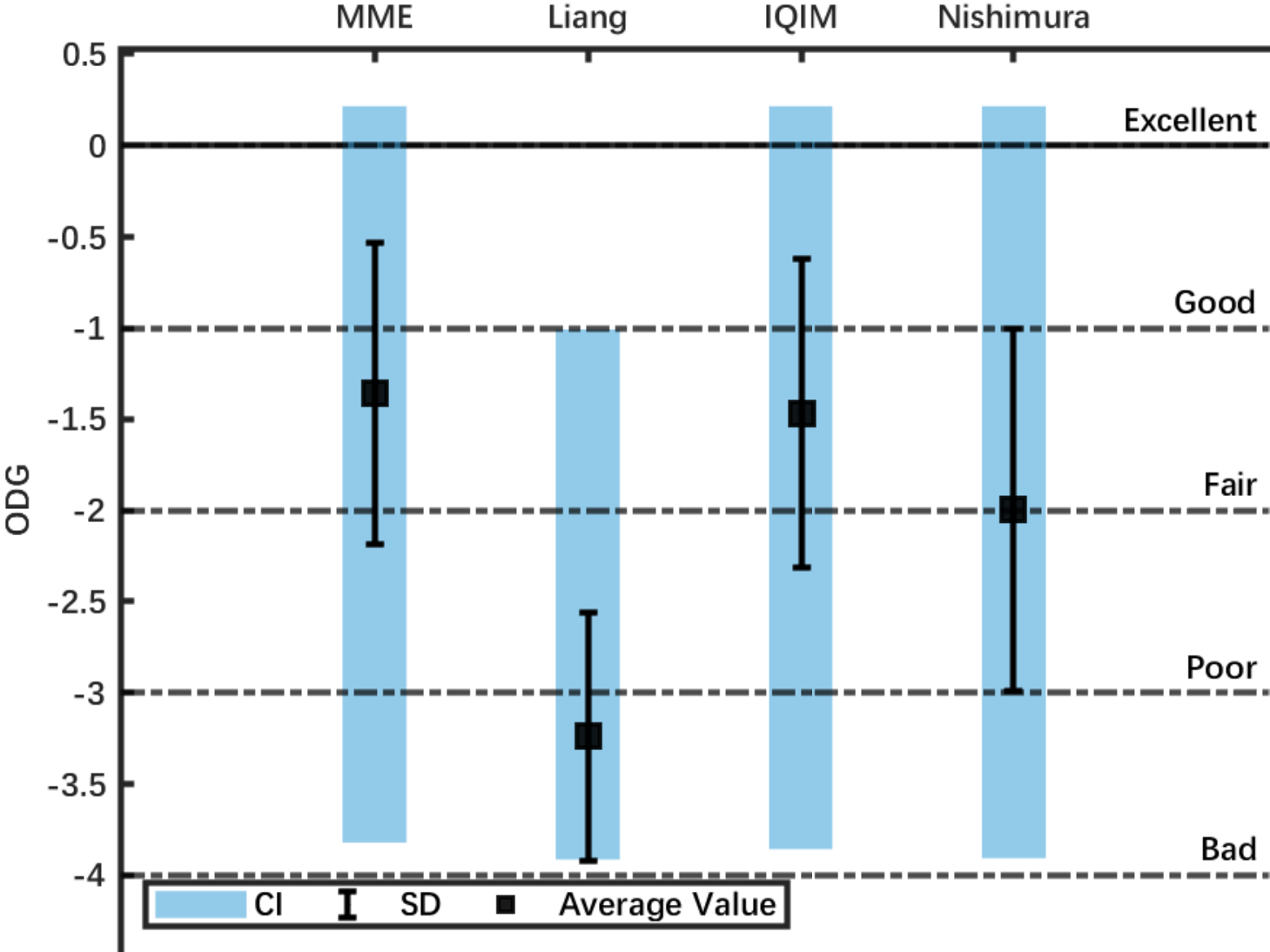}}
	
	\caption{The comparisons between methods in \cite{peng2011effective,LIANG2020107584,nishimura2013reversible} and the proposed MME with different metrics (a) SWR. (b)ODG.} 
	\label{Fig. experiment_1}
\end{figure}

Fig. \ref{Fig. experiment_1}(a) and Fig. \ref{Fig. experiment_1}(b) have respectively depicted the average value of SWR and ODG for the whole dataset.  The blue and black bars represent the confidence interval (CI)  and the range of standard deviation (SD). Fig. \ref{Fig. experiment_1}-(a) shows that the proposed method attains the highest SWR of $32.9891$dB, which is better than the $31.6769$dB of Liang, and the $28.3758$dB of IQIM. Moreover, the Nishimura method fails to work in this setting. Regarding ODG, the proposed method, IQIM and Nishimura work well, and they all belong the the class of QIM-based methods.

\begin{figure}[t!]
	\linespread{1}
	\centering
	\subfigure[]{\centering{\includegraphics[width=.35\textwidth]{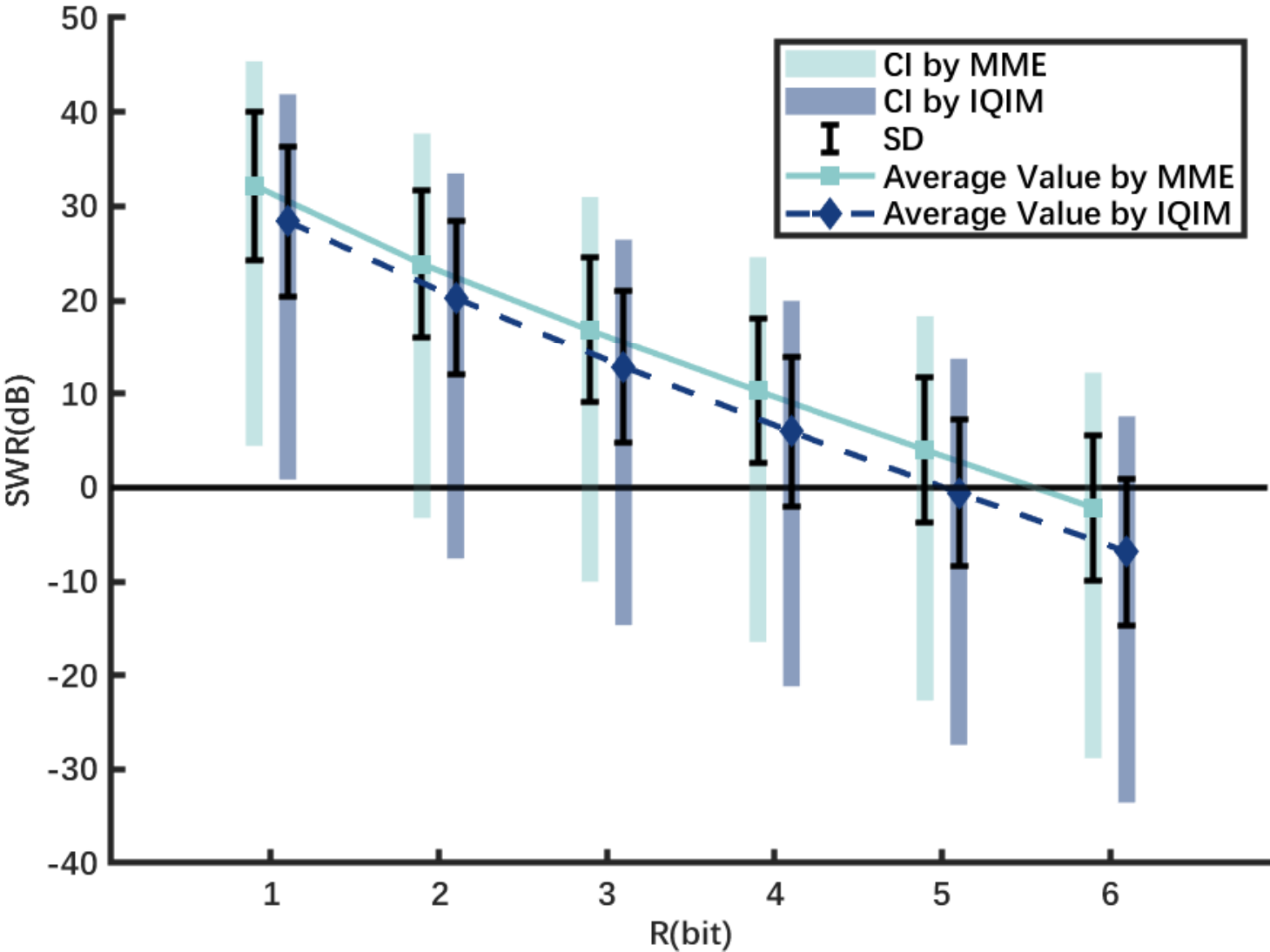}}}
	\quad
	\subfigure[]{\centering{\includegraphics[width=.5\textwidth]{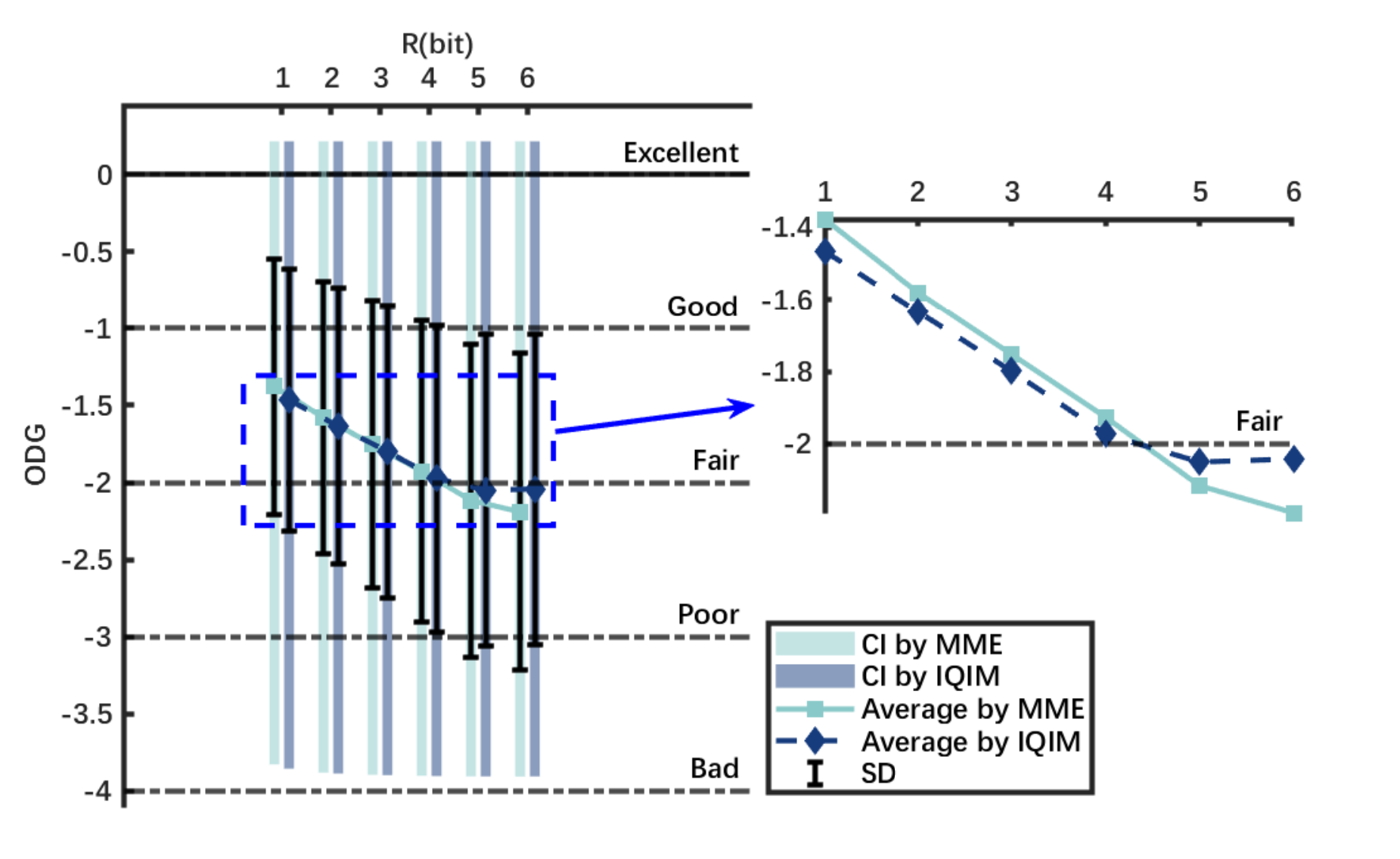}}}
	
	\caption{The impact of code rate $R$ of the performance of (a) SWR. (b)ODG.} 
	\label{Fig. experiment_1_2}
\end{figure}

More in-depth comparisons are made with the most related IQIM method hereby. 
Fig. \ref{Fig. experiment_1_2} shows the performance of SWR and ODG with different code rate $R$.
We have the following observations:
i) For the SWR performance, MME is strictly better than IQIM regardless of the code rate $R$.  
ii) In terms of ODG, when the performance is above the bar of ``fair'', MME generally outperforms IQIM, but falls short when both methods become unreliable.

The advantages of using more general lattices
for MME are further examined in Table \ref{ta_swr4}. In addition to the simple integer lattice $\mathbb{Z}$, the optimal lattices in $2$, $4$, and $8$ dimensions are respectively $A_2$, $D_4$ and $E_8$. The Conway-Sloane type labeling technique  \cite{conway2013sphere} can be used for these lattices.  Specifically, we have the following observations from the table. 
i)MME outperforms IQIM in all the metrics (SWR and ODG), regardless of the chosen type of lattices. For instance, the SWR of proposed lower bound in $A_2$ $(R=1)$ is $34.3771$dB, which is greater than the $28.3758$ SWR in IQIM.
ii) Choosing larger-dimensional optimal lattices is more desirable. Steady improvements can be observed by increasing the dimensions of the lattices from $\mathbb{Z}$ to $A_2$, $D_4$ and $E_8$.  
iii) Since the fine lattices are fixed to feature unit volumes while the coarse lattices grow as the code rate $R$ increases, the SWR/ODG decreases with $R$. 

\begin{table}[t!]
	\setlength{\tabcolsep}{1.5mm}
	\centering
	\renewcommand\arraystretch{1.5}
	\caption{The performance of MME with different lattices.}  
	\label{Tab. experiment 2}  
	\begin{tabular}{c c c c c c c }
		\hline
		\multirow{2}{*}{\textbf{Rate} $R$} & \multirow{2}{*}{\textbf{Metrics}} & \multicolumn{4}{c}{\textbf{MME with $\alpha=1-\frac{r_{pack}(\Lambda_f)}{r_{cov}(\Lambda_c)}$}} & \multicolumn{1}{c}{\multirow{2}{*}{\begin{tabular}[c]{@{}c@{}}\textbf{IQIM \cite{peng2011effective}} \end{tabular}}} \\ \cline{3-6}
		& & $\mathbb{Z}$ & $A_2$ & $D_4$ & $E_8$  & \multicolumn{1}{c}{}        
		\\ \hline
		\multicolumn{1}{c}{\multirow{4}{*}{1}} & $\alpha$  
		& 0.5 & 0.5670 & 0.6464 & 0.6464 
		& 0.5               
		\\ \cline{2-7}
		
		
		& SWR(dB)
		&34.5727	&34.3771	&34.2464	&36.0888	
		&28.3758
		
		\\ \cline{2-7} 
		
		& ODG 
		&-1.3185	&-1.3237	&-1.3316	&-1.2840	 
		&-1.4659
		\\ \hline
		
		\multicolumn{1}{c}{\multirow{4}{*}{2}} & $\alpha$
		&0.75&0.7835&0.8232&0.8232
		&0.75 
		\\ \cline{2-7}
		
		
		& SWR (dB)
		&25.2209	&25.75211	&26.3638	&28.3057	
		&20.2369
		\\ \cline{2-7} 
		
		& ODG
		&-1.5459	&-1.5290	&-1.5132	&-1.4723		
		&-1.6339
		\\ \hline
	\end{tabular}\label{ta_swr4}
\end{table}

\subsection{Robustness performance}
With the same setting as in Fig. \ref{Fig. experiment_1}-(a), we further investigate the robustness performance of the watermarking algorithms against AWGN attacks. Fig. \ref{BERfigure} has depicted the BER comparison for MME, Liang and IQIM. The Nishimura method has been excluded as its robustness performance is weak. 
The reason for  MME's better BER performance is that, by adjusting the scaling factor $\alpha$, it's watermarked vector can be tuned further away from the boarder of the decoding regions.

Fig. \ref{Fig. experiment_1_3}-(a) and  Fig. \ref{Fig. experiment_1_3}-(c) have drawn the GSNR for MME and IQIM for various SNR and $\alpha$. The GSNR of MME can go well beyond $10^{12}$, while IQIM flattens at about $3\times 10^{11}$. Fig. \ref{Fig. experiment_1_3}-(b) and Fig. \ref{Fig. experiment_1_3}-(d) are also attached to show the better GSNR performance of MME is not generated by scarifying the SWR performance.   
The simulation results above are remarkable: compared with other QIM-variants, MME can simultaneously achieve smaller distortion and better robustness. 

Since both the BER and SWR metrics are related to the $\Delta$ in MME, we plot the trade-off relations of BER versus SWR in Fig. \ref{RDfigure} by changing $\Delta$ and letting SNR=$25$dB. We have the following observations: i) The trade-off capability becomes better when $\Delta$ increases, where the curve rises smoothly when $\Delta=500$ and become sharper when $\Delta$ increases. ii) The ratio of SWR/BER decreases when $\Delta$ increases, which means larger distortion and higher robustness is achieved with a larger $\Delta$.

\begin{figure}[t!]
	\linespread{1}
	\centering{\includegraphics[width=.44\textwidth]{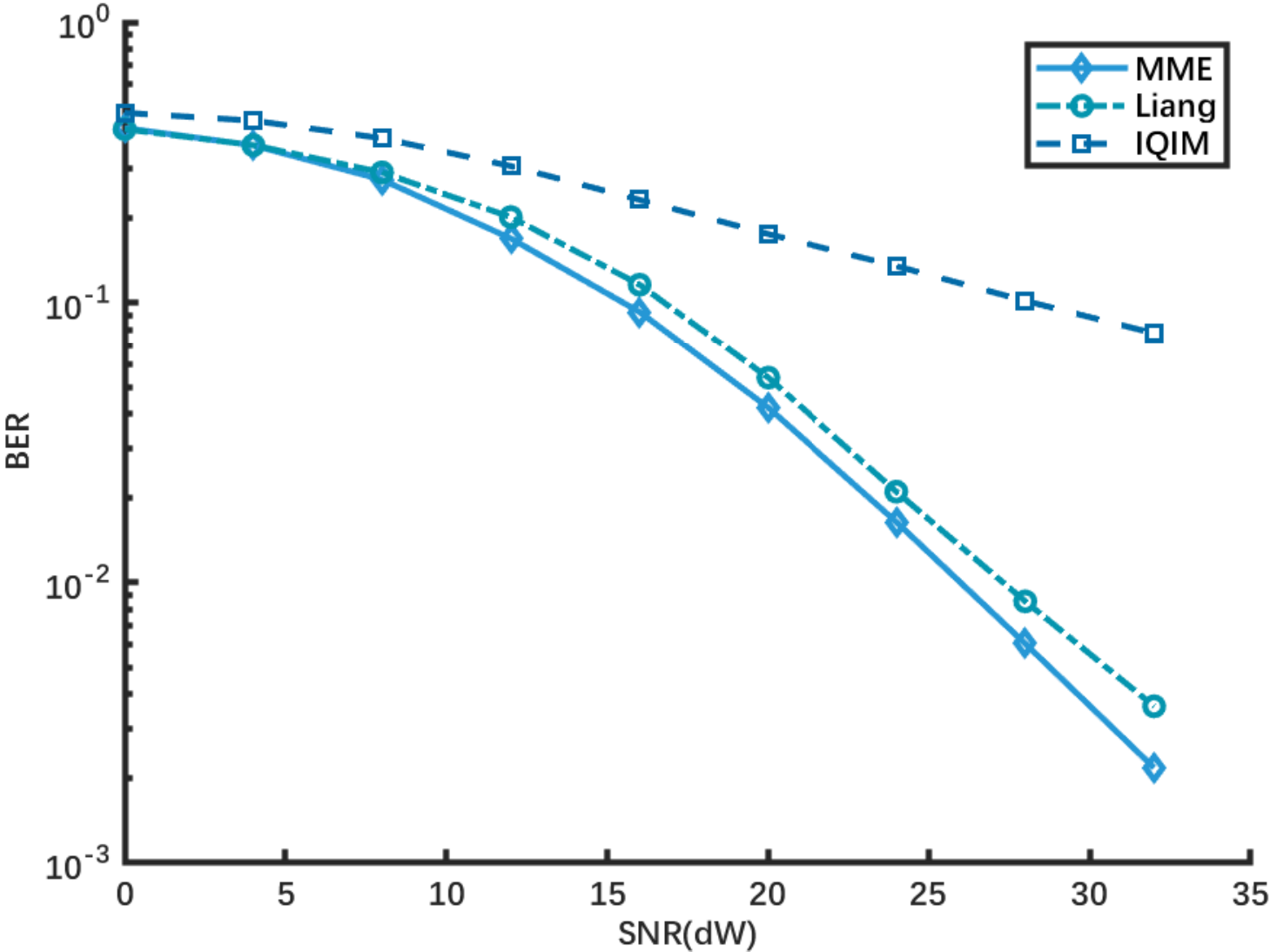}}
	\caption{BER versus SNR.} \label{BERfigure}
\end{figure}

\begin{figure}[t!]
	\linespread{1}
	\centering
	\subfigure[]{\centering{\includegraphics[width=.23\textwidth]{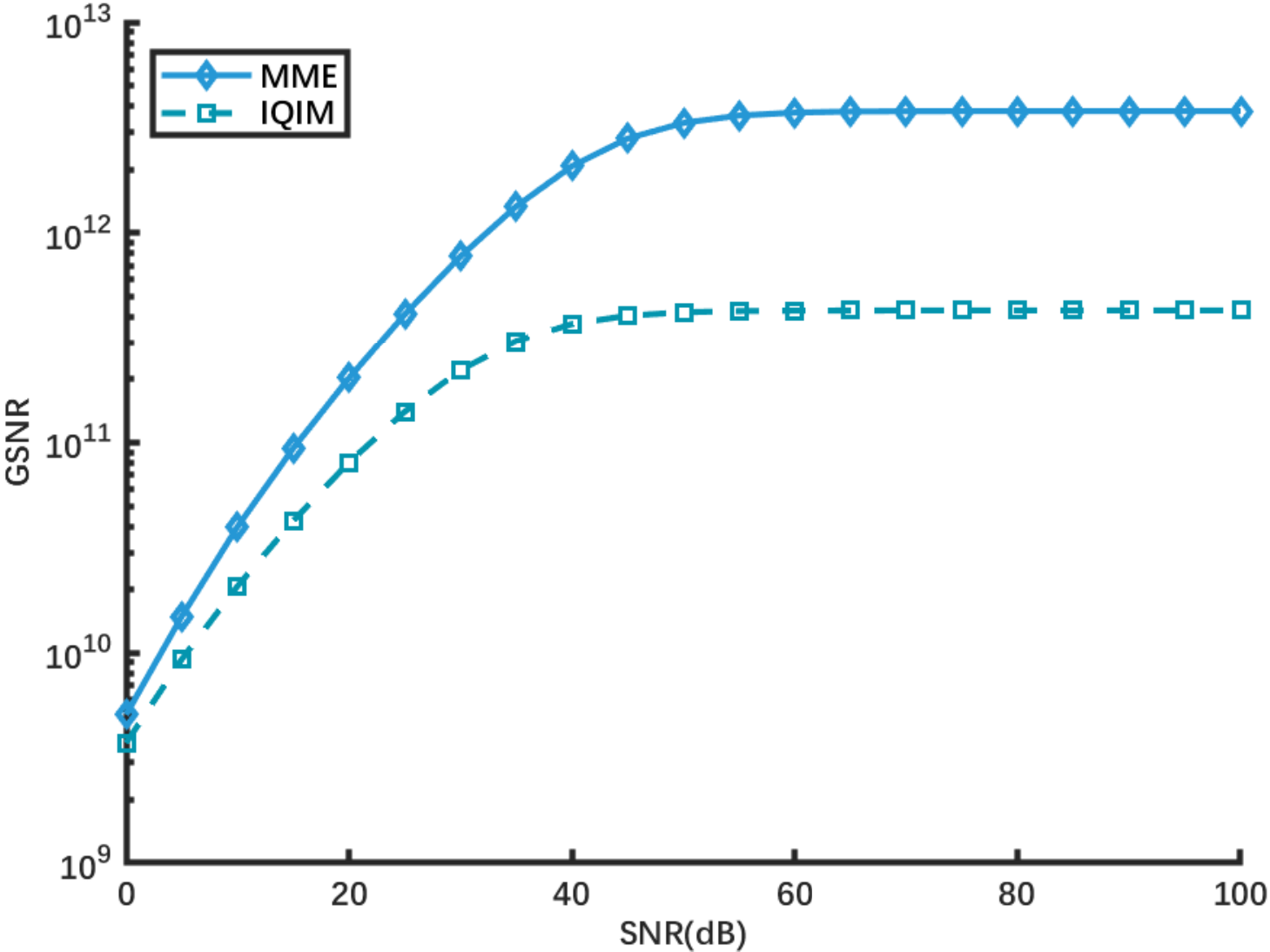}}}
	\subfigure[]{\centering{\includegraphics[width=.23\textwidth]{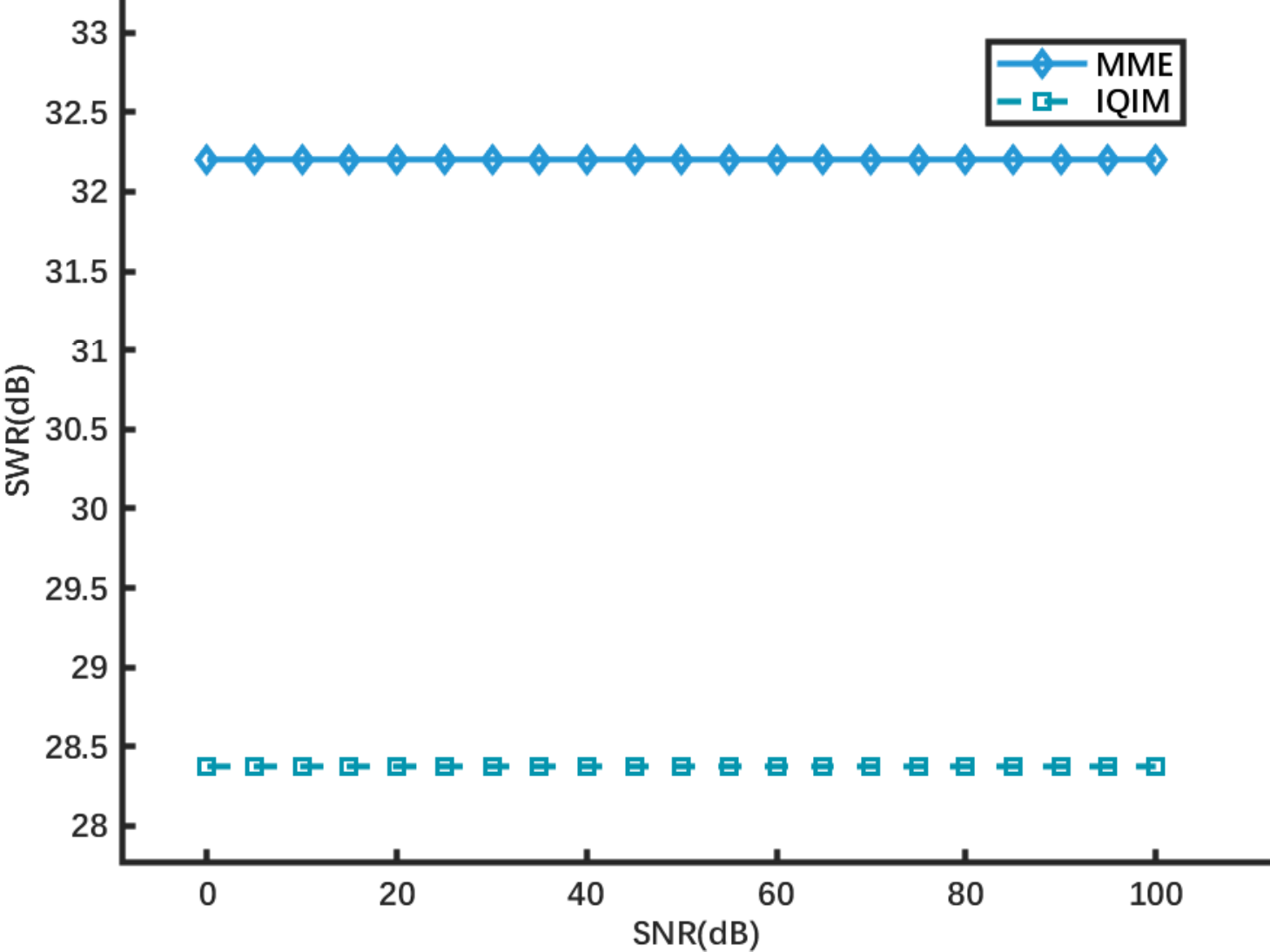}}}
	\subfigure[]{\centering{\includegraphics[width=.23\textwidth]{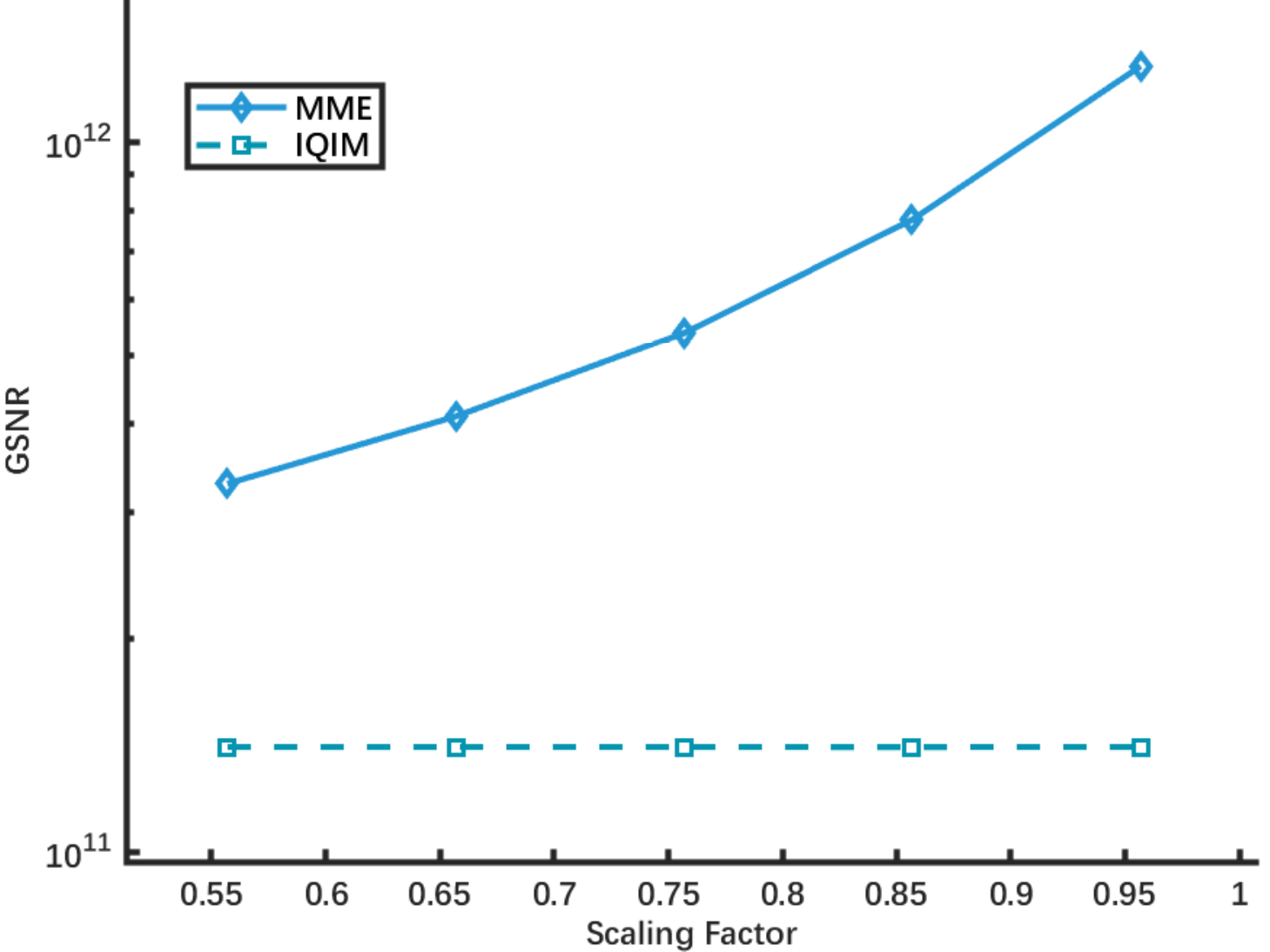}}}
	\subfigure[]{\centering{\includegraphics[width=.23\textwidth]{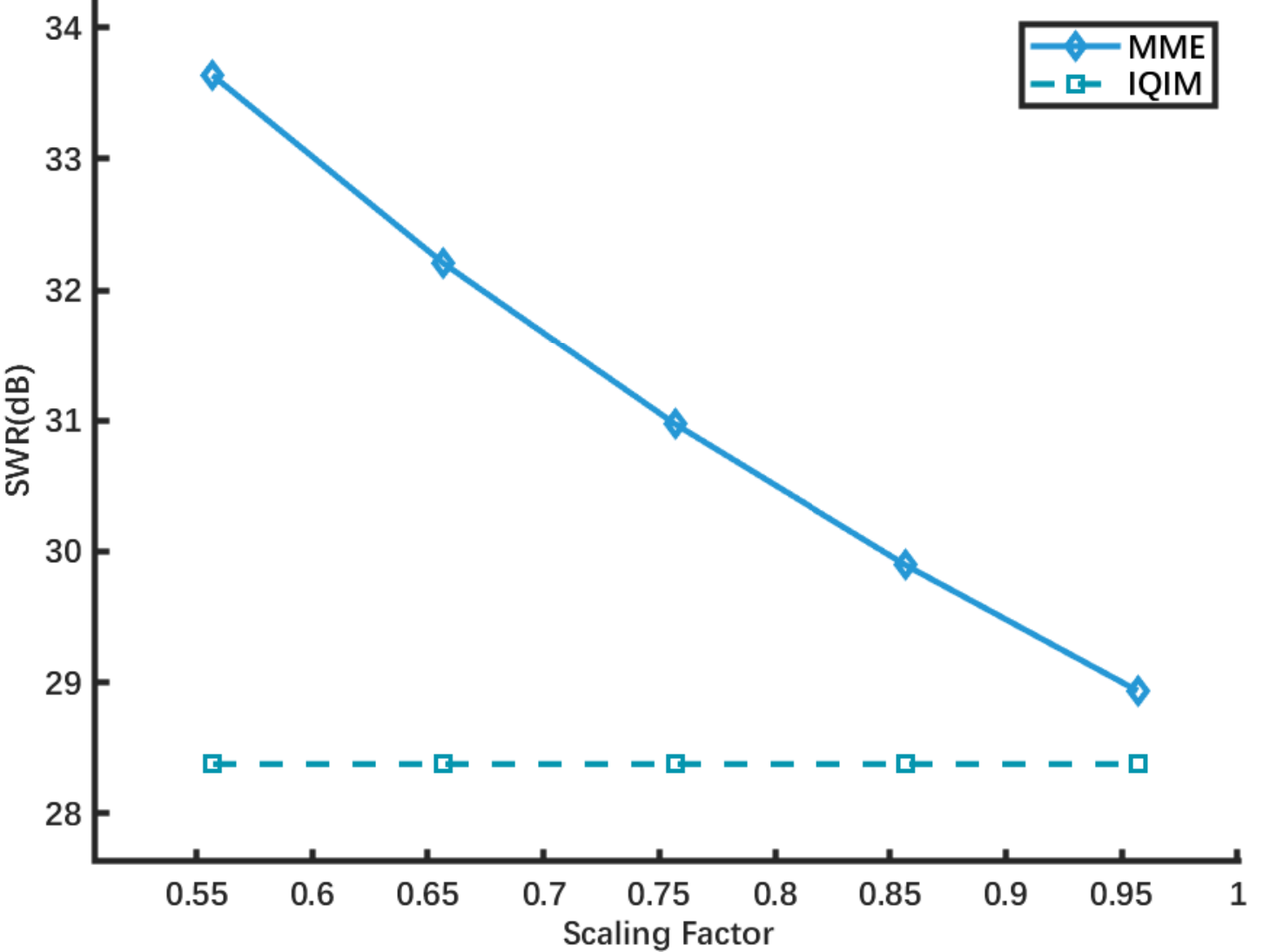}}}
	\caption{GSNR and SWR comparisons between MME and IQIM with different parameters. (a), (b) SNR ($\alpha=0.6569$, $R=1$). (c), (d) $\alpha$ (SNR$=25$, $R=1$).} 
	\label{Fig. experiment_1_3}
\end{figure}

\begin{figure}[t!]
	\linespread{1}
	\centering{\includegraphics[width=.44\textwidth]{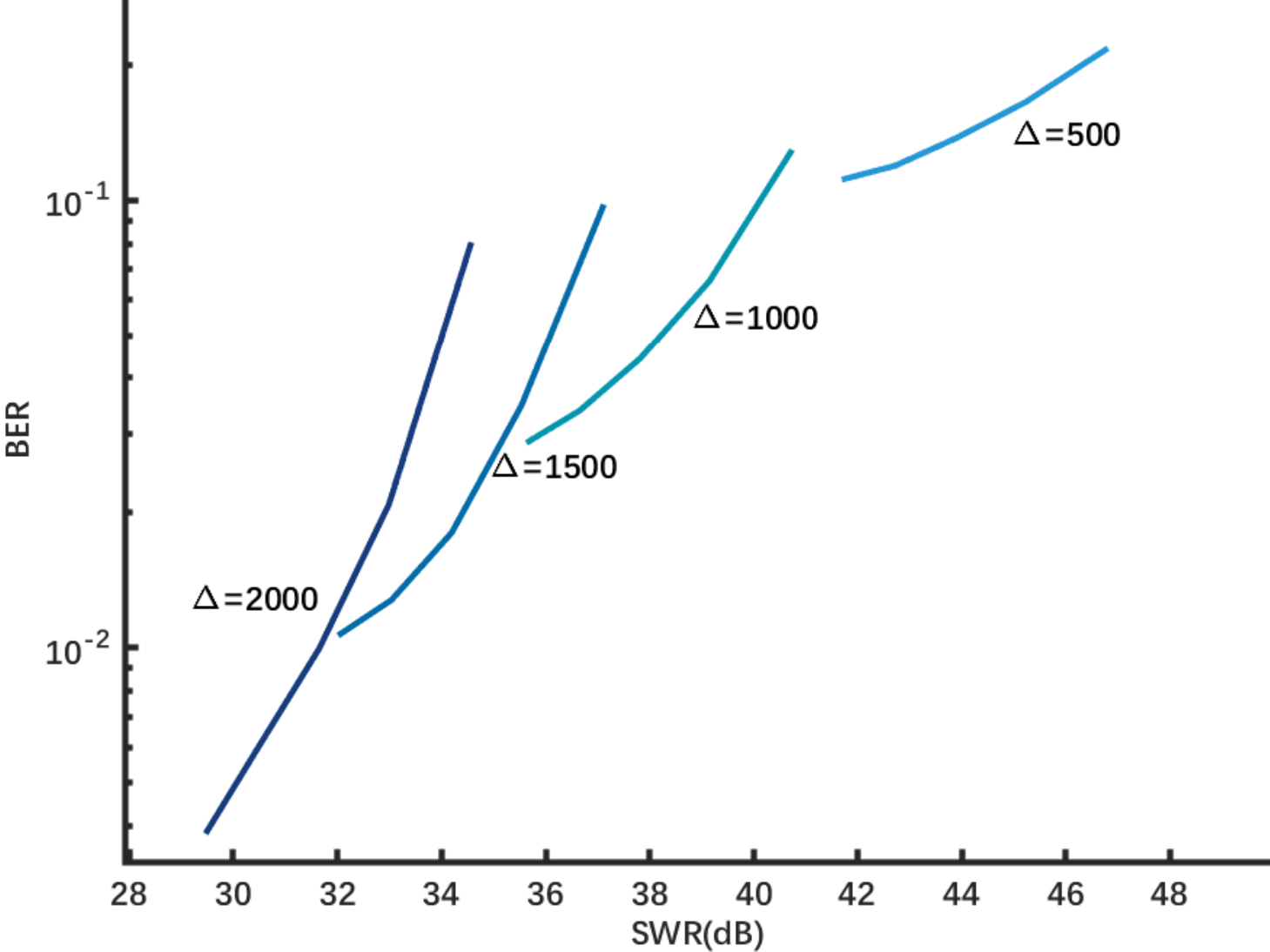}}
	\caption{BER versus SWR.} \label{RDfigure}
\end{figure}

\subsection{Reversibility}

\begin{figure*}[t!]
	\linespread{1}
	\centering
	\includegraphics[width=.9\textwidth]{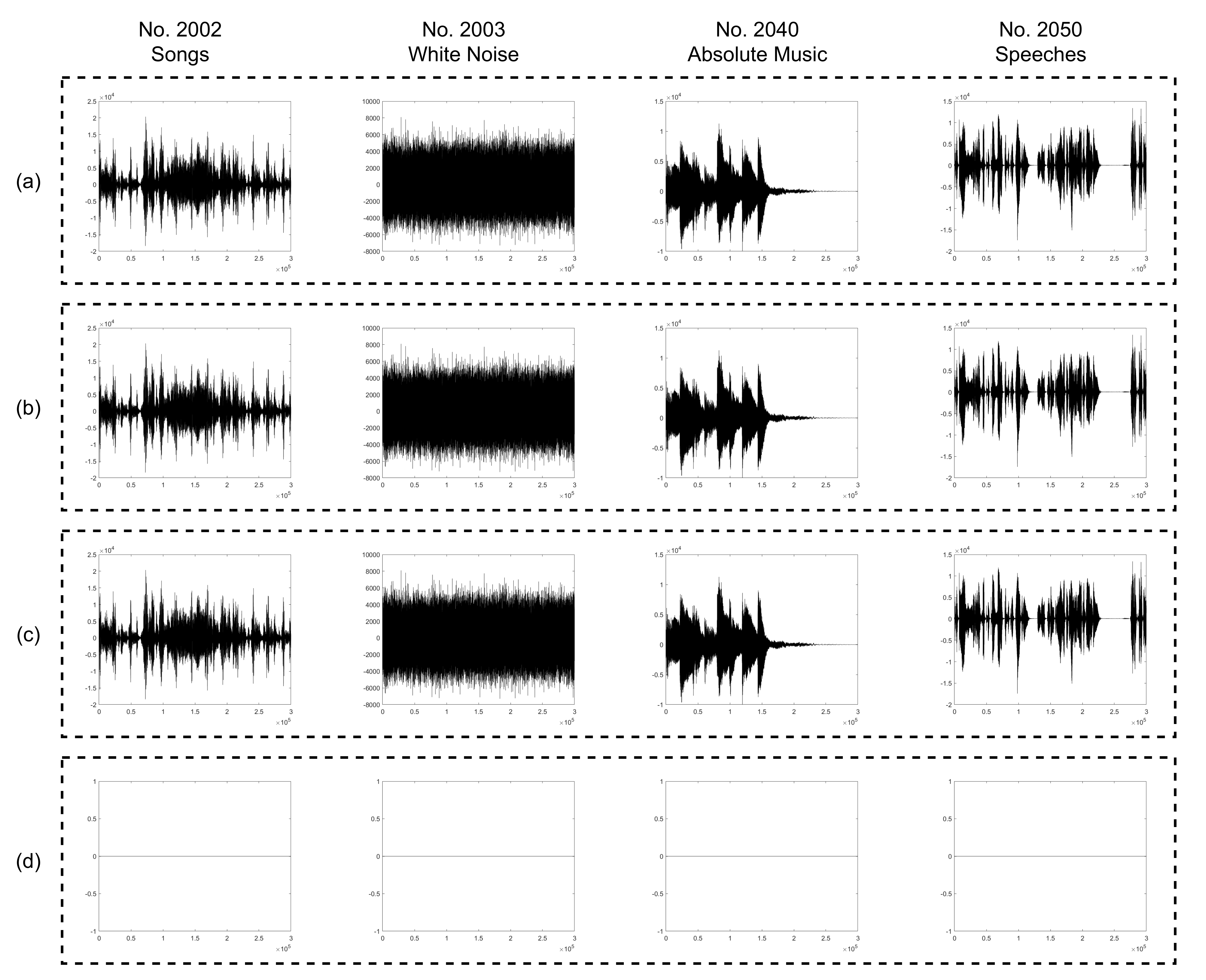}
	
	\caption{The signal waveforms at different stages of MME. 
		(a) Original signals. 
		(b) Watermarked signals by using MME, whose SWR are $38.0775$dB, $33.2487$dB, $34.0998$dB and	$35.1670$dB.
		(c) Restored signals.
		(d) The differences between the original signals and the restored signals.} 
	\label{Fig. experiment_4}
\end{figure*}

%

Due to the sensibility of HAS, the human awareness of sounds from different sources is also diverse, which may cause a distinct feeling for heterogeneous sounds.
To vividly demonstrate the actual effect of embedding, we draw the original, watermarked and restored signal of four typical audio host signals in Fig. \ref{Fig. experiment_4}. They are absolute music, songs, speeches and white noises, respectively. These host signals are labeled as No. 2002, 2003, 2040 and 2050.
Judging by the naked eye, the wave forms of the original and  the restored signals  are basically the same. In terms of watermarked signals, MME achieves the SWRs of  $38.0775$dB, $33.2487$dB, $34.0998$dB and $35.1670$dB for these host signals. Lastly, Fig. \ref{Fig. experiment_4}-(d) shows that are basically no difference between the original cover signals and the restored cover signals.

 
\section{Conclusions}
In this paper, a novel reversible audio watermarking method referred to as Meet-in-the-Middle Embedding (MME) has been proposed. We have justified its correctness, and make comparisons with existing methods, which show that MME exhibits less distortion and better robustness. Theoretical analyses on the SWR and GSNR have been presented.  Simulations show that the proposed method features a smaller amount of distortion and probability of decoding errors compared with   IQIM and other representative methods in RAW. 


%

\ifCLASSOPTIONcaptionsoff
  \newpage
\fi



%
\bibliographystyle{IEEEtranMine}
\bibliography{mybib}

%
%
%
%
%

\end{document}